\documentclass[conference,letterpaper]{IEEEtran}

\usepackage[numbers,sort&compress]{natbib}

\makeatletter

\usepackage{etex}

\usepackage{zref-savepos}

\newcounter{mnote}

\def\xmarginnote{%
  \xymarginnote{\hskip -\marginparsep \hskip -\marginparwidth}}

\def\ymarginnote{%
  \xymarginnote{\hskip\columnwidth \hskip\marginparsep}}

\long\def\xymarginnote#1#2{%
\vadjust{#1%
\smash{\hbox{{%
        \hsize\marginparwidth
        \@parboxrestore
        \@marginparreset
\footnotesize #2}}}}}

\def\mnoteson{%
\gdef\mnote##1{\refstepcounter{mnote}\label{##1}%
  \zsavepos{##1}%
  \ifnum20432158>\number\zposx{##1}%
  \xmarginnote{{\color{blue}\bf $\langle$\arabic{mnote}$\rangle$}}%
  \else
  \ymarginnote{{\color{blue}\bf $\langle$\arabic{mnote}$\rangle$}}%
  \fi%
}
  }
\gdef\mnotesoff{\gdef\mnote##1{}}
\mnoteson
\mnotesoff

\usepackage{framed}
\usepackage{comment}
\usepackage[svgnames,dvipsnames]{xcolor}
\usepackage[all]{xy}
\usepackage{tikz}
\usetikzlibrary{positioning,matrix,through,calc,arrows,fit,shapes,decorations.pathreplacing,decorations.markings,}

\tikzstyle{block} = [draw,fill=blue!20,minimum size=2em]

\usepackage{qsymbols,amssymb,mathrsfs}
\usepackage{amsmath}
\usepackage[standard,thmmarks]{ntheorem}
\theoremstyle{plain}
\theoremsymbol{\ensuremath{_\vartriangleleft}}
\theorembodyfont{\itshape}
\theoremheaderfont{\normalfont\bfseries}
\theoremseparator{}

\theoremstyle{nonumberplain}
\theoremheaderfont{\scshape}
\theorembodyfont{\normalfont}
\theoremsymbol{\ensuremath{_\blacktriangleleft}}

\theoremnumbering{arabic}
\theoremstyle{plain}
\usepackage{latexsym}
\theoremsymbol{\ensuremath{_\Box}}
\theorembodyfont{\itshape}
\theoremheaderfont{\normalfont\bfseries}
\theoremseparator{}

\theorembodyfont{\upshape}
\theoremprework{\bigskip\hrule}
\theorempostwork{\hrule\bigskip}

\usepackage[overload]{empheq} 

\let\iftwocolumn\if@twocolumn
\g@addto@macro\@twocolumntrue{\let\iftwocolumn\if@twocolumn}
\g@addto@macro\@twocolumnfalse{\let\iftwocolumn\if@twocolumn}

\mathtoolsset{showonlyrefs=false,showmanualtags}
\let\underbrace\LaTeXunderbrace 

\renewcommand{\eqref}[1]{\textup{(\refeq{#1})}} 
\newtagform{brackets}[]{(}{)}   
\usetagform{brackets}

\usepackage[Smaller]{cancel}

\PassOptionsToPackage{breaklinks,hyperindex=true,backref=false,bookmarksnumbered,bookmarksopen,linktocpage,colorlinks,linkcolor=BrickRed,citecolor=OliveGreen,urlcolor=Blue,pdfstartview=FitH}{hyperref}
\usepackage{hyperref}

\usepackage{graphicx,psfrag}
\graphicspath{{figure/}{image/}} 

\usepackage{diagbox} 

\usepackage{algorithm2e}
\usepackage{listings} 
\lstdefinelanguage{Maple}{
  morekeywords={proc,module,end, for,from,to,by,while,in,do,od
    ,if,elif,else,then,fi ,use,try,catch,finally}, sensitive,
  morecomment=[l]\#,
  morestring=[b]",morestring=[b]`}[keywords,comments,strings]
\DeclareMathAlphabet{\mathpzc}{OT1}{pzc}{m}{it}
\usepackage{upgreek} 
\usepackage{dsfont}  

\def\multi@nostar#1#2{%
  \expandafter\def\csname multi#1\endcsname##1{%
    \if ##1.\let\next=\relax \else
    \def\next{\csname multi#1\endcsname}     
    \expandafter\newcommand\csname #1##1\endcsname{#2}
    \fi\next}}

\def\multi@star#1#2{%
  \expandafter\def\csname #1\endcsname##1{#2}
  \multi@nostar{#1}{#2}
}

\newcommand{\multi}{%
  \@ifstar \multi@star \multi@nostar}

\multi*{rm}{\mathrm{#1}}
\multi*{mc}{\mathcal{#1}}
\multi*{op}{\mathop {\operator@font #1}}
\multi*{ds}{\mathds{#1}}
\multi*{set}{\mathcal{#1}}
\multi*{rsfs}{\mathscr{#1}}
\multi*{pz}{\mathpzc{#1}}
\multi*{M}{\boldsymbol{#1}}
\multi*{R}{\mathsf{#1}}
\multi*{RM}{\M{\R{#1}}}
\multi*{bb}{\mathbb{#1}}
\multi*{td}{\tilde{#1}}
\multi*{tR}{\tilde{\mathsf{#1}}}
\multi*{trM}{\tilde{\M{\R{#1}}}}
\multi*{tset}{\tilde{\mathcal{#1}}}
\multi*{tM}{\tilde{\M{#1}}}
\multi*{baM}{\bar{\M{#1}}}
\multi*{ol}{\overline{#1}}

\multirm  ABCDEFGHIJKLMNOPQRSTUVWXYZabcdefghijklmnopqrstuvwxyz.
\multiol  ABCDEFGHIJKLMNOPQRSTUVWXYZabcdefghijklmnopqrstuvwxyz.
\multitR   ABCDEFGHIJKLMNOPQRSTUVWXYZabcdefghijklmnopqrstuvwxyz.
\multitd   ABCDEFGHIJKLMNOPQRSTUVWXYZabcdefghijklmnopqrstuvwxyz.
\multitset ABCDEFGHIJKLMNOPQRSTUVWXYZabcdefghijklmnopqrstuvwxyz.
\multitM   ABCDEFGHIJKLMNOPQRSTUVWXYZabcdefghijklmnopqrstuvwxyz.
\multibaM   ABCDEFGHIJKLMNOPQRSTUVWXYZabcdefghijklmnopqrstuvwxyz.
\multitrM   ABCDEFGHIJKLMNOPQRSTUVWXYZabcdefghijklmnopqrstuvwxyz.
\multimc   ABCDEFGHIJKLMNOPQRSTUVWXYZabcdefghijklmnopqrstuvwxyz.
\multiop   ABCDEFGHIJKLMNOPQRSTUVWXYZabcdefghijklmnopqrstuvwxyz.
\multids   ABCDEFGHIJKLMNOPQRSTUVWXYZabcdefghijklmnopqrstuvwxyz.
\multiset  ABCDEFGHIJKLMNOPQRSTUVWXYZabcdefghijklmnopqrstuvwxyz.
\multirsfs ABCDEFGHIJKLMNOPQRSTUVWXYZabcdefghijklmnopqrstuvwxyz.
\multipz   ABCDEFGHIJKLMNOPQRSTUVWXYZabcdefghijklmnopqrstuvwxyz.
\multiM    ABCDEFGHIJKLMNOPQRSTUVWXYZabcdefghijklmnopqrstuvwxyz.
\multiR    ABCDEFGHIJKL NO QR TUVWXYZabcd fghijklmnopqrstuvwxyz.
\multibb   ABCDEFGHIJKLMNOPQRSTUVWXYZabcdefghijklmnopqrstuvwxyz.
\multiRM   ABCDEFGHIJKLMNOPQRSTUVWXYZabcdefghijklmnopqrstuvwxyz.

\newcommand{\dotleq}{\buildrel \textstyle  .\over {\smash{\lower
      .2ex\hbox{\ensuremath\leqslant}}\vphantom{=}}}
\newcommand{\dotgeq}{\buildrel \textstyle  .\over {\smash{\lower
      .2ex\hbox{\ensuremath\geqslant}}\vphantom{=}}}

\newcommand{\bM}{\begin{bmatrix}}
\newcommand{\eM}{\end{bmatrix}}
\newcommand{\bSM}{\left[\begin{smallmatrix}}
\newcommand{\eSM}{\end{smallmatrix}\right]}
\renewcommand*\env@matrix[1][*\c@MaxMatrixCols c]{%
  \hskip -\arraycolsep
  \let\@ifnextchar\new@ifnextchar
  \array{#1}}

\newqsymbol{`N}{\mathbb{N}}
\newqsymbol{`R}{\mathbb{R}}
\newqsymbol{`Z}{\mathbb{Z}}

\newqsymbol{`|}{\mid}
\newqsymbol{`8}{\infty}
\newqsymbol{`1}{\left}
\newqsymbol{`2}{\right}
\newqsymbol{`6}{\partial}
\newqsymbol{`0}{\emptyset}
\newqsymbol{`-}{\leftrightarrow}

\newcommand{\sgn}{\operatorname{sgn}}

\DeclareFontFamily{OMX}{MnSymbolE}{}
\DeclareSymbolFont{MnLargeSymbols}{OMX}{MnSymbolE}{m}{n}
\SetSymbolFont{MnLargeSymbols}{bold}{OMX}{MnSymbolE}{b}{n}
\DeclareFontShape{OMX}{MnSymbolE}{m}{n}{
    <-6>  MnSymbolE5
   <6-7>  MnSymbolE6
   <7-8>  MnSymbolE7
   <8-9>  MnSymbolE8
   <9-10> MnSymbolE9
  <10-12> MnSymbolE10
  <12->   MnSymbolE12
}{}
\DeclareFontShape{OMX}{MnSymbolE}{b}{n}{
    <-6>  MnSymbolE-Bold5
   <6-7>  MnSymbolE-Bold6
   <7-8>  MnSymbolE-Bold7
   <8-9>  MnSymbolE-Bold8
   <9-10> MnSymbolE-Bold9
  <10-12> MnSymbolE-Bold10
  <12->   MnSymbolE-Bold12
}{}

\let\llangle\@undefined
\let\rrangle\@undefined
\DeclareMathDelimiter{\llangle}{\mathopen} {MnLargeSymbols}{'164}{MnLargeSymbols}{'164}
\DeclareMathDelimiter{\rrangle}{\mathclose} {MnLargeSymbols}{'171}{MnLargeSymbols}{'171}

\newcommand{\imod}[1]{\allowbreak\mkern10mu({\operator@font mod}\,\,#1)}

\newcommand{\threecols}[3]{
\hbox to \textwidth{%
      \normalfont\rlap{\parbox[b]{\textwidth}{\raggedright#1\strut}}%
        \hss\parbox[b]{\textwidth}{\centering#2\strut}\hss
        \llap{\parbox[b]{\textwidth}{\raggedleft#3\strut}}%
    }
}

\newcommand{\reason}[2][\relax]{
  \ifthenelse{\equal{#1}{\relax}}{
    \left(\text{#2}\right)
  }{
    \left(\parbox{#1}{\raggedright #2}\right)
  }
}

\newcommand{\utag}[2]{\mathop{#2}\limits^{\text{(#1)}}}
\newcommand{\uref}[1]{(#1)}

\let\SavedDoubleVert\relax
\let\protect\relax
{\catcode`\|=\active
  \xdef\extendvert{\protect\expandafter\noexpand\csname extendvert \endcsname}
  \expandafter\gdef\csname extendvert \endcsname#1{\mskip-5mu \left.%
      \ifx\SavedDoubleVert\relax \let\SavedDoubleVert\|\fi
     \:{\let\|\SetDoubleVert
       \mathcode`\|32768\let|\SetVert
     #1}\:\right.\mskip-5mu}
}
\def\SetVert{\@ifnextchar|{\|\@gobble}
    {\egroup\;\mid@vertical\;\bgroup}}
\def\SetDoubleVert{\egroup\;\mid@dblvertical\;\bgroup}

\begingroup
 \edef\@tempa{\meaning\middle}
 \edef\@tempb{\string\middle}
\expandafter \endgroup \ifx\@tempa\@tempb
 \def\mid@vertical{\middle|}
 \def\mid@dblvertical{\middle\SavedDoubleVert}
\else
 \def\mid@vertical{\mskip1mu\vrule\mskip1mu}
 \def\mid@dblvertical{\mskip1mu\vrule\mskip2.5mu\vrule\mskip1mu}
\fi

\makeatother

\usepackage{fouridx}
\usepackage{framed}
\usetikzlibrary{positioning,matrix}

\usepackage{paralist}
\usepackage{enumerate}

\usepackage[normalem]{ulem}

{\endMakeFramed}

\newenvironment{ybox}{
	\setlength{\FrameSep}{1.5mm}
	\setlength{\FrameRule}{0mm}
  \MakeFramed {\FrameRestore}}%
{\endMakeFramed}

\newenvironment{gbox}{
	\setlength{\FrameSep}{1.5mm}
\setlength{\FrameRule}{0mm}
  \MakeFramed {\FrameRestore}}%
{\endMakeFramed}

{\endMakeFramed}

 {\endMakeFramed}

\usepackage{enumitem}

\usepackage{etoolbox}

\let\theparentequation\theequation
\patchcmd{\theparentequation}{equation}{parentequation}{}{}

\renewenvironment{subequations}[1][]{
	\refstepcounter{equation}%
	\setcounter{parentequation}{\value{equation}}
	\setcounter{equation}{0}
	\def\theequation{\theparentequation\alph{equation}}%
	\let\parentlabel\label
	\ifx\\#1\\\relax\else\label{#1}\fi
	\ignorespaces
}{%
	\setcounter{equation}{\value{parentequation}}
	\ignorespacesafterend
}

\newcommand*{\nextParentEquation}[1][]{
	\refstepcounter{parentequation}
	\setcounter{equation}{0}
	\ifx\\#1\\\relax\else\parentlabel{#1}\fi
}

\DeclareMathOperator{\rank}{rank}

\newcommand{\wskc}{C_{\op{W}}}
\newcommand{\skc}{C_{\op{S}}}
\newcommand{\pkc}{C_{\op{P}}}
\newcommand{\rco}{R_{\op{CO}}}
\newcommand{\rl}{R_{\op{L}}}

\newcommand{\Fq}{\mathbb{F}_q}

\title{Secret Key Agreement and Secure Omniscience  \\of  Tree-PIN Source with Linear Wiretapper}
\author{Praneeth Kumar Vippathalla, Chung Chan, Navin Kashyap and Qiaoqiao Zhou
	\thanks{C.\ Chan (email: chung.chan@cityu.edu.hk) is with the Department of Computer Science, City University of Hong Kong. His work is supported by a grant from the University Grants Committee of the Hong Kong Special Administrative Region, China (Project No. 21203318).}
    \thanks{Q.\ Zhou (email: zq115@ie.cuhk.edu.hk) is with the Institute of Network Coding and the Department of Information Engineering, The Chinese University of Hong Kong.}
	\thanks{N.\ Kashyap (nkashyap@iisc.ac.in) and Praneeth Kumar V.\ (praneethv@iisc.ac.in) are with the Department of Electrical Communication Engineering, Indian Institute of Science, Bangalore 560012. Their work was supported in part by a Swarnajayanti Fellowship awarded to N.\ Kashyap by the Department of Science \& Technology (DST), Government of India.}
	}

\begin{document}
\newif\ifPAGELIMIT
\PAGELIMITfalse
\IEEEoverridecommandlockouts

\maketitle
\begin{abstract}
  While the \emph{wiretap secret key capacity} remains unknown for general source models even in the two-user case, we obtained a single-letter characterization for a large class of multi-user source models with a \emph{linear wiretapper} who can observe any linear combinations of the source. We introduced the idea of \emph{irreducible sources} to show existence of an optimal communication scheme that achieves perfect omniscience with minimum leakage of information to the wiretapper. This implies a duality between the problems of wiretap secret key agreement and \emph{secure omniscience}, and such duality potentially holds for more general sources.

\end{abstract} 

\section{Introduction} \label{sec:introduction}
The problem of multiterminal secret key agreement was studied by Csisz\'{a}r and Narayan in \cite{csiszar04}. They derived the single-letter expression for the secret key capacity $\skc$ when the wiretapper has no side information. Remarkably, they established a duality between the problem of  secret key agreement and the problem of communication for omniscience by the users, which means that attaining omniscience by users is enough to extract a secret key of maximum rate. However, the characterization of secret key capacity when the wiretapper has side information $\wskc$ was left open, and only gave some upper bounds on it. Later Gohari and Anantharam, in \cite{aminsource},  provided strengthened upper bounds and lower bounds. Furthermore, they proved a duality between secret key agreement with wiretapper side information and the problem of communication for omniscience by a neutral observer, where the neutral observer attains omniscience instead of the users. But this equivalence does not give an exact single-letter characterization of $\wskc$. Nevertheless in some special cases, it is known exactly. In particular,
\cite{alireza19} studied  a pairwise independent network (PIN) source model defined on trees  with wiretapper side information obtained by passing the edge random variables through independent channels. For this model, $\wskc$ was characterized using the conditional minimum rate of communication for omniscience characterization given in \cite{csiszar04}, and provided a scheme that achieves it. The final form of $\wskc$ is similar to that of $\skc$ except for the conditioning with respect to wiretap side information. One can see that the linear wiretapper case is not covered by this model. 

Recently in \cite{chan20secure}, Chan et al. have studied the problem of secure omniscience in the context of multiterminal  secure information exchange, and  explored its duality connection to the problem of wiretap secret key agreement. In the secure omniscience problem, every user tries to attain omniscience by communicating interactively using their private observations from a correlated source, however, with a goal to minimize the information leakage to the wiretapper who has side information about the source. Interestingly, in the case of finite linear source (FLS) involving two active users and a wiretapper, they provided an explicit characterization of the wiretap secret key capacity and the minimum leakage rate for omniscience $\rl$. In fact, the achievable communication scheme for wiretap secret key capacity involves secure omniscience. Motivated by this result, they conjectured that such a duality holds for the entire class of FLS. In this  paper, we address this question and completely resolves it in the subclass of tree-PIN model but with a linear wiretapper, which is the most general wiretapper in the class of FLS. 

The PIN sources have received a wide attention in the secret key agreement problem without wiretapper side information, see \cite{sirinpin,chan19,qiao20}.  The main motivation for studying PIN sources is that they model the problem of generating a global key out of locally generated keys by user pairs. In the study of general PIN sources, the subclass of tree-PIN sources play an important role. For the tree-PIN \cite{sirinpin}, secret key capacity is achieved by using a linear and non-interactive communication scheme that  propagates a key across the tree. This protocol indeed serves as a building block in the tree-packing protocol for the  general PIN. It was proved in \cite{chan19} that the tree-packing protocol is even optimal for the constrained secrecy capacity $\skc(R)$ where $R$ is the total discussion rate. The optimality was shown by deriving a matching converse bound. Recently, \cite{qiao20} identified a large class of PIN models where the tree-packing protocol achieves the entire rate region where each point is a  tuple of achievable  key rate and individual discussion rates.

A problem that is closely related to secure omniscience is the coded cooperative data exchange (CCDE) problem with secrecy constraint, see for e.g., \cite{sprinston13, courtade16}. The problem of CCDE considers a hypergraphical source and studies the one-shot omniscience. The hypergraphical model generalizes the PIN model within the class of FLS. \cite{courtade16} studied the secret key agreement in the CCDE context and characterized the number of transmissions required versus the number of SKs generated. Whereas \cite{sprinston13} considered the same model but with wiretapper side information and explored the leakage aspect of an omniscience protocol. However, the security notion considered therein does not allow the eavesdropper to recover even one hyperedge (data packet) of the source from the communication except what is already available. But the communication scheme can still reveal information about the source. In this paper we are interested to minimize the leakage of the total information to the wiretapper. Though we consider the asymptotic notion, the  designed optimal communication scheme uses only finite number of realizations of the source. Hence this scheme can find application even  in CCDE problems.
 
In this paper, for a tree-PIN with linear wiretapper, we exactly characterize $\rl$ and $\wskc$ by giving an optimal linear (non-interactive) communication scheme. To do this, we first  reduce the source to an irreducible source and then we give a communication protocol that achieves both perfect omniscience, a notion that was introduced in  \cite{sirinperfect}, and perfect alignment with wiretapper. In perfect omniscience, terminals recover the source perfectly using only a finite number of source realizations. Moreover, perfect alignment means that the wiretapper observations can be completely recovered from the communication alone. Ideally, it should be the other way around - the communication should be completely recoverable from the eavesdropper's observations, so that the eavesdropper learns nothing new about the source. However, it may not always be possible to design a communication for omniscience satisfying this requirement. So, we ask for the next best thing, that a large part of the communication contains information already known to the eavesdropper.

The paper is organized as follows. We introduce the problem and notations in section \ref{sec:problem}. Section \ref{sec:results} contains the main results whereas the proofs are presented in section \ref{sec:proof}. Section \ref{sec:prot} focuses on an explicit secure omniscience protocol. Finally, we conclude with possible future directions and open problems in section \ref{sec:conc}.

\section{Problem formulation}\label{sec:problem}
In this section, we describe two different scenarios in the context of multiterminal setting where  the terminals communicate publicly using their correlated observations to perform a task securely from the eavesdropper, who has access to the public communication along with side information.  More precisely, let $V=[m]:=\left\lbrace1, \ldots, m\right\rbrace$ be the set of users  and $\opw$ denotes the wiretapper.  Let  $\RZ_1,\ldots \RZ_m$ and $\RZ_{\opw}$ be the random variables  taking values in finite alphabets $\mc{Z}_1,\ldots \mc{Z}_m$ and $\mc{Z}_{\opw}$ respectively, and their joint distribution is given by $P_{\RZ_1 \ldots \RZ_m \RZ_{\opw}}$. Let $\RZ_V := (\RZ_i: i \in V)$ and $\RZ_i^n$ denote the $n$ i.i.d. realizations  of $\RZ_i$.  Each user has access to the corresponding random variable. Upon observing $n$ i.i.d. realizations, the terminals communicate interactively using their observations and possibly independent private randomness on the noiseless and authenticated channel. In other words, the communication made by an user in any round depends on all the previous rounds communication and  user's observations.  Let $\RF^{(n)}$ denotes this interactive communication. We say $\RF^{(n)}$ is \emph{non-interactive}, if it is of the form $(\tRF_i^{(n)}: i \in V)$, where $\tRF_i^{(n)}$ depends on only on $\RZ_i^n$ and the  private randomness of user $i$. Note that the eavesdropper has access to the pair $(\RF^{(n)}, \RZ_{\opw}^n)$. At the end of the communication, users output a value in a finite set using their observations and $\RF^{(n)}$. For example, user $i$ outputs $\RE_i^{(n)}$ using $(\RF^{(n)}, \RZ_i^n)$ and its private randomness. See Fig.~\ref{fig:system}.
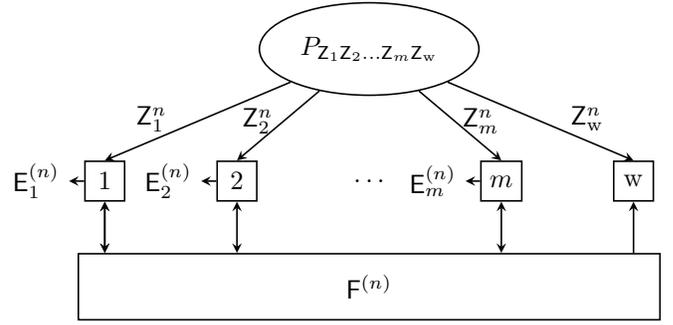
\begin{figure}
\centering
\begin{tikzpicture}[>=stealth,semithick, node distance = 5 em]
\tikzstyle{source}=[ellipse, draw, semithick, minimum height=3.5 em, minimum width=7.5 em];
\tikzstyle{user}=[rectangle, draw, semithick, minimum height=1.5 em, minimum width=1.5 em];
\tikzstyle{chnl}=[rectangle, draw, semithick, minimum height=2.5 em, minimum width=22 em];
\node   (dots)   {$\ldots$};
\node[source]      (source)       [above of = dots] {$P_{\RZ_1\RZ_2\ldots \RZ_m\RZ_{\opw}}$};
\node[user]      (2)       [left of = dots] {$2$};
\node[user]      (1)       [left of = 2] {$1$};
\node[user]      (m)       [right of = dots] {$m$};
\node[user]      (w)       [right of = m] {${\opw}$};
\node[chnl]      (chnl)       [below of =dots, yshift=1 em] {$\RF^{(n)}$};
\node (k1) [left of = 1, xshift=2.4 em] {$\RE_1^{(n)}$};
\node (k2) [left of = 2, xshift=2.4 em] {$\RE_2^{(n)}$};
\node (km) [left of = m, xshift=2.4 em] {$\RE_m^{(n)}$};
\draw [->] (source) edge node[above,near end] {$\RZ_1^n$} (1.north);
\draw [->] (source) edge node[above,near end]{$\RZ_2^n$}(2.north);
\draw [->] (source) edge node[above,near end]{$\RZ_m^n$} (m.north);
\draw [->] (source) edge node[above,near end]{$\RZ_{\opw}^n$} (w.north);
\draw [<->]  (1.south)--(1|-chnl.north);
\draw [<->]  (2.south)--(2|-chnl.north);
\draw [<->]  (m.south)--(m|-chnl.north);
\draw [->]  (w|-chnl.north)--(w.south);
\draw [<->]  (1.south)--(1|-chnl.north);
\draw [->]  (1.west)--(k1.east);
\draw [->]  (2.west)--(k2.east);
\draw [->]  (m.west)--(km.east);
\end{tikzpicture}
\caption{Multiterminal source model with wiretapper side information. The terminals interactively discuss over a public channel using their observations from a correlated source to agree upon a common randomness which must be kept secure from the wiretapper.}
\label{fig:system}
 \end{figure}
\subsection{Secure Omniscience}
In the secure omniscience scenario, each user tries to recover the observations of the other users except wiretapper's. We say that $(\RF^{(n)}, \RE_1^{(n)}, \ldots, \RE_m^{(n)})_ {n \geq 1}$  is an omniscience scheme if it satisfies the recoverability condition for omniscience
\begin{align}\label{eq:omn:recoverability}
\liminf_{n \to \infty} \Pr(\RE_1^{(n)} = \ldots =\RE_m^{(n)} = \RZ_V^n) = 1.
\end{align}
The minimum leakage rate for omniscience is defined as 
\begin{align}
\begin{split}
 \rl&:= \inf  \biggl\lbrace \limsup_{n \to \infty} \frac{1}{n}I(\RF^{(n)} \wedge \RZ_V^n|\RZ_{\opw}^n) \biggr\rbrace \label{eq:rl}
 \end{split}
\end{align}
where the infimum is over all omniscience schemes. We sometimes use $\rl(\RZ_V||\RZ_{\opw})$ instead of $\rl$ to make the source explicit. When there is no wiretapper side information, then the above notion coincides with the minimum rate of communication for omniscience, $\rco$ \cite{csiszar04}. And the conditional minimum rate of communication for omniscience, $\rco(\RZ_V|\RJ)$, is used  in the case when all the users have the shared randomness $\RJ^n$ along with  their private observations. This means that user $i$ observes $(\RJ^n, \RZ_i^n)$. 

\subsection{Secret Key Agreement}
In the secure secret key agreement, each user tries to recover a common randomness that is kept secure from the wiretapper. Specifically, we say that $(\RF^{(n)}, \RE_1^{(n)}, \ldots, \RE_m^{(n)})_ {n \geq 1}$  is a secret key agreement (SKA) scheme if there exists a sequence $(\RK^{(n)})_{n \geq 1}$  such that
\begin{subequations}
\label{eq:sk:constraints}
\begin{align}
\liminf_{n \to \infty} \Pr(\RE_1^{(n)} = \ldots =\RE_m^{(n)} = \RK^n) = 1 \label{eq:sk:recoverability},\\
\limsup_{n \to \infty}\left[\log |\mc{K}^{(n)}| - I(\RK^{(n)}\wedge \RF^{(n)},\RZ_{\opw}^n)\right] =0\label{eq:sk:secrecy},
\end{align}
where \eqref{eq:sk:recoverability} is the key recoverability condition and \eqref{eq:sk:secrecy} is the secrecy condition of the key and $|\mc{K}^{(n)}|$ denotes the cardinality of the range of $\RK^{(n)}$.
\end{subequations}
The wiretap secret key capacity  is defined as 
\begin{align}
 \wskc:= \sup \left\lbrace \liminf_{n \to \infty} \frac{1}{n} \log |\mc{K}^{(n)}| \label{eq:wskc}\right\rbrace
\end{align}
where the supremum is over all SKA schemes. The quantity $\wskc$ is also sometimes written as $\wskc(\RZ_V||\RZ_{\opw})$. In \eqref{eq:wskc}, we use $\skc$ instead of $\wskc$, when the wiretap side information is set to a constant.   Similarly, we use $\pkc(\RZ_V| \RJ)$  in the case when wiretap side information is  $\RZ_{\opw}= \RJ$ and all the users have the shared random variable $\RJ$ along with  their private observations $\RZ_i$. The quantities $\skc$ and $\pkc(\RZ_V|\RJ)$ are referred to  as secret key capacity of  $\RZ_V$ and private key capacity of $\RZ_V$ with compromised-helper side information $\RJ$ respectively.
\subsection{Tree PIN source with linear wirtapper}
 A source $\RZ_V$ is said to be \emph{Tree-PIN} if there exists a tree $T=(V,E,\xi)$ and for each edge $e \in E$, there is a non-negative integer $n_e$ and a random vector $\RY_e = \left( \RX_{e,1}, \ldots, \RX_{e,n_e} \right)$. We assume that the collection of random variables $\RX :=(\RX_{e,k}: e\in E, k \in [n_e])$ are i.i.d. and each component is  uniformly distributed over a finite field, say $\Fq$. For $i \in V$,
 \begin{align*}
  \RZ_i = \left( \RY_e : i \in \xi (e) \right) .
 \end{align*} 

 The linear wiretapper's side information $\RZ_{\opw}$ is defined as 
\begin{align*}
 \RZ_{\opw} = \RX \MW,
\end{align*}
where $\RX$ is a $1 \times (\sum_{e \in E}n_e)$ vector and $\MW$ is a $(\sum_{e \in E}n_e) \times n_w$ full column-rank matrix over $\Fq$. We sometimes refer to $\RX$ as the base vector. We refer to the pair $(\RZ_V, \RZ_{\opw})$ defined as above as the \emph{Tree-PIN source with linear wiretapper}. This is a special case of finite linear sources \cite{chan10} where both $\RZ_V$ and $\RZ_{\opw}$ can be written as $\RX\MM$ and $\RX\MW$ respectively for some matrices $\MM$ and $\MW$. In the context of FLS, we say a communication scheme $\RF^{(n)}$ is \emph{linear}, if each user's communication is a linear function of its observations and the previous communication on the channel. Without loss of generality, linear communication can also be assumed to be non-interactive.  In the rest of the paper, we consider only matrices over $\Fq$ unless otherwise specified.
\subsection{Motivating example} The following example of a tree-PIN source with linear wiretapper appeared in our earlier work \cite{chan20secure}, where we constructed an optimal secure omniscience scheme. Let $V=\{1,2,3,4\}$ and  
     \begin{align}
        \RZ_{\opw} &= \RX_a+\RX_b+\RX_c, \\
        \RZ_1 &= \RX_a, \RZ_2 = (\RX_a, \RX_b), \RZ_3 = (\RX_b, \RX_c), \RZ_4 =  \RX_c,
      \end{align}
  where $\RX_a$, $\RX_b$ and $\RX_c$ are uniformly random and independent bits. The tree here is a path of  length $3$ (Fig.~\ref{fig:exampletree}) and the wiretapper observes the linear combination of all the edge random variables. For secure omniscience, terminals 2 and 3, using $n=2$ i.i.d. realizations of the source, communicate linear combinations of their observations. The communication is of the form, $\RF^{(2)} =(\tRF_2^{(2)},\tRF_3^{(2)})$, where  $\tRF_2^{(2)} =\RX^2_a+\MM \RX^2_b$ and $\tRF_3^{(2)}=(\MM + \MI) \RX_b^2 +\RX_c^2$ with $\MM:=\bM 1 & 1\\ 1 & 0\eM$.  Since the matrices $\MM$ and $\MM+\MI$ are invertible, all the terminals can recover $\RZ_V^2$ using this communication. For example, user 1 can first  recover $\RX_b^2$ from $(\RX_a^2, \tRF_2^{(2)})$ as $\RX_b^2 = (\MM+\MI)(\RX_a^2+ \tRF_2^{(2)})$, then $\RX_b^2$ can be used along with $\tRF_3^{(2)}$ to recover $\RX_c^2$ as $\RX_c^2 = (\MM+\MI)\RX_b^2+ \tRF_3^{(2)}$.  More interestingly, this communication is aligned with the eavesdropper's observations, since $\RZ^2_{\opw} = \tRF_2^{(2)}+\tRF_3^{(2)}$. 
  
  For minimizing leakage, this kind of alignment must happen. For example, if $\RZ^2_{\opw}$ were not contained in the span of $\tRF_2^{(2)}$ and $\tRF_3^{(2)}$, then the wiretapper could infer a lot more from the communication.  Ideally if one wants zero leakage, then $\RF^{(n)}$ must be within the span of $\RZ^n_{\opw}$, which is not feasible in many cases because with that condition, the communication might not achieve omniscience in the first place. Therefore keeping this in mind, it is reasonable to assume that there can be components of $\RF^{(n)}$ outside the span of $\RZ^n_{\opw}$. And we look for communication schemes which span as much of $\RZ_{\opw}$ as possible. Such an alignment condition is used to control the leakage. In this particular example, it turned out that an omniscience communication that achieves $\rco$ can be made to completely align with the wiretapper side information.  With the motivation from this example, we in fact showed that such an alignment phenomenon holds true in the entire class of tree-PIN with linear wiretapper.

\section{Main results}\label{sec:results}
The following two propositions give upper and lower bounds on minimum leakage rate for a general source $(\RZ_V,\RZ_{\opw})$. The lower bound on $\rl$ in terms of wiretap secret key capacity is obtained by using the idea of  privacy amplification on the recovered source.
While the multi-letter upper bound is given in terms of any communication made using first $n$ i.i.d. realizations.

\begin{proposition}[\cite{chan20secure}, Theorem 1]\label{thm:RL:lb}
    For the secure omniscience scenario with $|V|\geq 2$,
      \begin{align}
       \rl &\geq H(\RZ_V|\RZ_{\opw}) - \wskc.\label{eq:RL:lb}
        \end{align}
\end{proposition}
\begin{proposition}[\cite{chan20secure}, Theorem 2]
  \label{thm:RL:ub}
    For the secure omniscience scenario,
    \begin{align}
        \rl &\leq \frac{1}{n} [\rco(\RZ_V^n|\RF^{(n)}) + I(\RZ_V^n\wedge \RF^{(n)} | \RZ_{\opw}^n)] \leq \rco \label{eq:RL:ub},
    \end{align}
    where the  inequality holds for any integer $n$ and valid public discussion $\RF^{(n)}$ for block length $n$.
\end{proposition}

Before we present our result, we will discuss some notions related to G\'{a}cs-K\"{o}rner common information, which play an important role in proving the result. The G\'{a}cs-K\"{o}rner common information of  $\RX$ and $\RY$ with joint distribution $P_{\RX,\RY}$ is defined as 
\begin{align}\label{eq:gk}
 J_{\op{GK}}(\RX,\RY) := \max \left\lbrace H(\RG) : H(\RG|\RX)=H(\RG|\RY) =0 \right\rbrace
\end{align}
A $\RG$  that satisfies the constraint in \eqref{eq:gk} is called a common function (c.f.) of $\RX$ and $\RY$. An optimal $\RG$ in \eqref{eq:gk} is called a \emph{maximal common function} (m.c.f.) of $\RX$ and $\RY$, and is denoted by $\op{mcf}(\RX, \RY)$. Similarly, for $n$ random variables, $\RX_1, \RX_2, \ldots, \RX_n$,  we can extend these definitions by replacing the condition in \eqref{eq:gk} with $H(\RG|\RX_1)=H(\RG|\RX_2)=\ldots=H(\RG|\RX_n)=0$. For a finite linear source pair $(\RZ_1, \RZ_2)$, i.e., $\RZ_1 = \RX \MM_1$ and $\RZ_2=\RX \MM_2$ for some matrices $\MM_1$ and $\MM_2$ where $\RX$ is a $1 \times n$ row vector that is  uniformly distributed on $\Fq^n$, it was shown in \cite{chan18zero} that the $\op{mcf}(\RZ_1, \RZ_2)$ is a linear function of $\RX$ which means that there exists a matrix $\MM_g$ such that $\op{mcf}(\RZ_1, \RZ_2) = \RX \MM_g$. 

The main result of this paper is the following theorem.
\begin{theorem}\label{thm:cwsk:red} 
For a Tree-PIN source $\RZ_V$ with linear wiretapper observing $\RZ_{\opw}$,
\begin{align*}
\wskc &= \min_{e \in E} H(\RY_e|\op{mcf}(\RY_e,\RZ_{\opw})),  \\ 
\rl &=\left(\sum_{e \in E}n_e -n_w\right)\log_2q -\wskc \text{ bits}.
\end{align*}
In fact, a linear non-interactive scheme  is  sufficient to achieve both $\wskc$ and $\rl$ simultaneously.
\end{theorem}

The above theorem shows that  the intrinsic upper bound on $\wskc$ holds with equality. In the multiterminal setting, the intrinsic bound that follows from \cite[Theorem 4]{csiszar04} is given by 
\begin{align*}
 \wskc(\RZ_V ||\RZ_{\opw}) \leq \min_{\RJ-\RZ_{\opw}-\RZ_V}\pkc(\RZ_V|\RJ).
\end{align*}
 This is analogous to the intrinsic bound for the two terminal case \cite{maurer99intrinsic}. 
For the class of tree-PIN sources with linear wiretapper,  when $\RJ^*= \left( \op{mcf}(\RY_e,\RZ_{\opw}) \right)_{e \in E}$, it can be shown that  $\pkc(\RZ_V|\RJ^*)= \min_{e \in E} H(\RY_e|\op{mcf}(\RY_e,\RZ_{\opw})) $. This can be derived using the characterization  in \cite{csiszar04} of the conditional minimum rate of communication for omniscience, $\rco(\RZ_V|\RJ^*)$. In fact, the same derivation can also be found in \cite{alireza19} for a  $\RJ$ that is obtained by passing edge random variables through independent channels. In particular, $\RJ^{*}$ is a function of edge random variables $(\RY_e)_{e \in E}$ because $\op{mcf}(\RY_e, \RZ_{\opw})$ is a function of $\RY_e$. Therefore, we can see that $\pkc(\RZ_V|\RJ^*)$, which is an upper bound on $ \min_{\RJ-\RZ_{\opw}-\RZ_V}\pkc(\RZ_V|\RJ)$, matches with the $\wskc$ obtained from Theorem~\ref{thm:cwsk:red}.

 Furthermore, the theorem guarantees that in the tree-PIN case with linear wiretapper, we can achieve the wiretap secret key capacity through a linear secure omniscience scheme. This shows that omniscience can be useful even beyond the case when there is no wiretapper side information, where  \cite{csiszar04} showed that achieving omniscience is enough for the terminals to achieve the secret key capacity.
 
 Our proof of Theorem~\ref{thm:cwsk:red} is through a reduction to the particular subclass of \emph{irreducible} sources, which we defined next.
 
\begin{definition}
 A Tree-PIN source with linear wiretapper is said to be \emph{irreducible} iff $\op{mcf}(\RY_e, \RZ_{\opw}) $ is a constant function for every edge $e \in E$ .
\end{definition}

Whenever there is an edge $e$ such that $\RG_e:=\op{mcf}(\RY_e, \RZ_{\opw}) $ is a non-constant function, the user corresponding to a vertex incident on $e$ can reveal $\RG_e$ to the other users. This communication does not leak any additional information to the wiretapper, because $\RG_e$ is a function of $\RZ_{\opw}$. Intuitively, for the further communication, $\RG_e$ is not useful and hence can be removed from the source. After the reduction the m.c.f. corresponding to $e$ becomes a constant function. In fact, we can carry out the reduction until the source becomes irreducible. This idea of reduction is illustrated through the following example.

\begin{Example}
Let us consider a source $\RZ_V$ defined on a path of length 3, which is shown in Fig.~\ref{fig:exampletree}. Let $\RY_a = (\RX_{a1}, \RX_{a2})$, $\RY_b = \RX_{b1}$ and $\RY_c =\RX_{c1}$, where $\RX_{a1}$, $\RX_{a2}$, $\RX_{b1}$ and $\RX_{c1}$ are uniformly random and independent bits. 
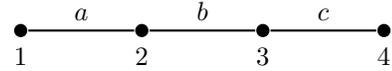
\begin{figure}[h]
\centering
\begin{tikzpicture}[-,>=stealth,thick, auto]
\tikzstyle{vertex}=[circle,fill=black,inner sep=0pt,minimum size=5pt];

\node[vertex]      (1)        [label= below:{$1$}] {};
\node[vertex]      (2)  [right  = 4 em of 1,label= below:{$2$}] {};
\node[vertex]      (3)  [right  = 4 em of 2,label= below:{$3$}] {};
\node[vertex]      (4)  [right  = 4 em of 3,label= below:{$4$}] {};

\draw (1) -- (2) node [midway, above] {$a$};
\draw (2) -- (3) node [midway, above] {$b$};
\draw (3) -- (4) node [midway, above] {$c$};
\end{tikzpicture}
\caption{A path of length 3}
\label{fig:exampletree}
 \end{figure}
 If $\RZ_{\opw}=\RX_{b1}+\RX_{c1}$, then the source is irreducible because $\op{mcf}(\RY_e, \RZ_{\opw})$ is a constant function for all $e \in \{a,b,c\}$. 
 
 However if  $\RZ_{\opw}=(\RX_{a1}+\RX_{a2},  \RX_{b1}+\RX_{c1})$, then the source is not irreducible, as $\op{mcf}(\RY_a, \RZ_{\opw}) =\RX_{a1}+\RX_{a2}$, which is a non-constant function. An equivalent representation of the source is
 $\RY_a = (\RX_{a1}, \RG_{a})$, $\RY_b = \RX_{b1}$, $\RY_c =\RX_{c1}$ and $\RZ_{\opw}=(\RG_{a}, \RX_{b1}+\RX_{c1})$, where $\RG_{a}=\RX_{a1}+\RX_{a2}$, which is also a uniform bit independent of $(\RX_{a1}, \RX_{b1}, \RX_{c1})$. So, for omniscience, user 2 initially can reveal $\RG_{a}$ without affecting the information leakage as it is completely aligned to $\RZ_{\opw}$. Since everyone has $\RG_a$, users can just communicate according to the omniscience scheme corresponding to the source without $\RG_a$. Note that this new source is irreducible.
\end{Example}

The next lemma shows that the kind of reduction to an irreducible source used in the above example is indeed optimal in terms of $R_L$ and $\wskc$ for all tree-PIN sources with linear wiretapper.

\begin{lemma}\label{lem:irred} 
 If the Tree-PIN source with linear wiretapper $(\RZ_V,\RZ_{\opw})$ is not irreducible then there exists an irreducible source $(\tRZ_V, \tRZ_{\opw})$ such that 
 \begin{align*}
\wskc(\RZ_V|| \RZ_{\opw}) = \wskc(\tRZ_V||\tRZ_{\opw}),\\ \rl(\RZ_V||\RZ_{\opw}) = \rl(\tRZ_V||\tRZ_{\opw}),\\
H(\RY_e|\op{mcf}(\RY_e,\RZ_{\opw})) = H(\tRY_e),
\end{align*}
for all $e \in E$.
\end{lemma}

As a consequence of Lemma 1, to prove Theorem 1, it suffices to consider only irreducible sources. For ease of reference, we re-state the theorem for irreducible sources below.

\begin{theorem}\label{thm:cwsk:irred} 
If Tree-PIN source with linear wiretapper is irreducible then 
\begin{align*}
\wskc &= \min_{e \in E} H(\RY_e)=\skc, \\
\rl &=\left(\sum_{e \in E}n_e -n_w\right)\log_2q- \skc\text{ bits},
\end{align*}
where $\skc$ is the secret key capacity of Tree-PIN source without the wiretapper side information~\cite{csiszar04}.
\end{theorem}

\section{Proofs}\label{sec:proof}
\ifPAGELIMIT
In this section we provide the essential proof ideas while the full proofs are available in the longer version~\cite{treepin21}. \subsection{Proof sketch of Lemma~\ref{lem:irred}}
First we identify an edge $e$ such that $\RG_e:=\op{mcf}(\RY_e, \RZ_{\opw})$ is a non-constant function. Then, by appropriately transforming the random vector $\RY_e$, we can separate out $\RG_e$ from the random variables corresponding to the edge and the wiretapper. Later we argue that the source $(\RZ_V,\RZ_{\opw})$ can be reduced into $(\tRZ_V,\tRZ_{\opw})$ by removing $\RG_e$ entirely without affecting  $\wskc$ and $\rl$. And we repeat this process until the source becomes irreducible. At each stage, to show that  $\op{mcf}(\tRY_b, \tRZ_{\opw})=\op{mcf}(\RY_b, \RZ_{\opw})$, for $b \neq e$, and  $\op{mcf}(\tRY_e, \tRZ_{\opw})$ is a constant function, we  use  the following lemma which is proved in \cite{treepin21}.
\begin{lemma}\label{lem:indgk}
 If $(\RX,\RY)$ is independent of $\RZ$, then  $\op{mcf}(\RX, (\RY,\RZ)) = \op{mcf}(\RX,\RY)$ and $\op{mcf}((\RX,\RZ), (\RY,\RZ)) = (\op{mcf}(\RX,\RY) ,\RZ)$.
\end{lemma}

\subsection{Proof sketch of Theorem~\ref{thm:cwsk:irred}}
\emph{Converse part.} An  upper bound on  $\wskc$  is $\skc$, because no wiretapper side information can only increase the key generation ability of users. It was shown in \cite[Example 5]{csiszar04} that if the random variables of a source form a Markov chain on a tree, then $\skc = \min_{(i,j) : \{i,j\} = \xi(e) } I(\RZ_i ; \RZ_j)$. In the tree-PIN case, which satisfies the Markov property, this turns out to be $\skc=\min_{e \in E} H(\RY_e)$. As consequence, we have $\wskc \leq \min_{e \in E} H(\RY_e)$ and 
\begin{align}\label{eq:rl:conv}
\begin{split}
 \rl&\utag{a}\geq H(\RZ_V|\RZ_{\opw}) -\wskc \\ 
 &\utag{b}= \left(\sum_{e \in E}n_e -n_w\right)\log_2q -\wskc \\
 &\geq \left(\sum_{e \in E}n_e -n_w\right)\log_2q  - \min_{e \in E} H(\RY_e)
 \end{split}
\end{align}\\
where \uref{a} follows from Proposition~\ref{thm:RL:lb} and \uref{b} is due to the full column-rank assumption on $\MW$.

\emph{Achievability part.}  In this section, we will show the existence of an omniscience scheme with leakage rate $\left(\sum_{e \in E}n_e -n_w\right)\log_2q  - \min_{e \in E} H(\RY_e)$. Hence $\rl \leq \left(\sum_{e \in E}n_e -n_w\right)\log_2q  - \min_{e \in E} H(\RY_e)$,  which together with the chain of inequalities~\eqref{eq:rl:conv} imply that $\wskc = \min_{e \in E} H(\RY_e)=\skc$  and $\rl =\left(\sum_{e \in E}n_e -n_w\right)\log_2q- \skc$. In particular, for achieving a secret key of rate $\wskc = \min_{e \in E} H(\RY_e)$, the terminals use privacy amplification on the recovered source.

In fact, the existence of an omniscience scheme is shown by first constructing a template for the communication with desired properties and then showing the existence of an instance of it by random coding. The  following are the key components involved in this construction.
\begin{enumerate}
\item \emph{Deterministic scheme:} A scheme is said to be deterministic if  terminals are  not allowed to use any locally generated private randomness. 
 \item \emph{Perfect omniscience~\cite{sirinperfect}:} For a fixed $n \in \bb{N}$, $\RF^{(n)}$ is said to achieve perfect omniscience if  terminals can recover the source $\RZ_V^n$ perfectly, i.e., $H(\RZ_V^n|\RF^{(n)}, \RZ_i^n) =0$ for all $i \in V$. If we do not allow any private randomness, then $H( \RF^{(n)} |  \RZ_V^n) = 0$, which implies
  \begin{align*}\label{eq:perfectomni}
  \begin{split}
   \frac{1}{n}  I(\RZ_V^n\wedge \RF^{(n)} | \RZ_{\opw}^n) &= \frac{1}{n}\left [H( \RF^{(n)} | \RZ_{\opw}^n) - H( \RF^{(n)} | \RZ_{\opw}^n, \RZ_V^n) \right] \\&= \frac{1}{n}H( \RF^{(n)} | \RZ_{\opw}^n).
  \end{split}
\end{align*} 
\item \emph{Perfect alignment:} For an $n \in \bb{N}$, we say that $\RF^{(n)}$ perfectly aligns with $\RZ_{\opw}^n$ if $H( \RZ_{\opw}^n|\RF^{(n)} ) = 0$. Note that $\RZ_{\opw}^n$ is only recoverable from $\RF^{(n)}$ but not the other way around. In this case, $H( \RF^{(n)} | \RZ_{\opw}^n) =H( \RF^{(n)}) - H( \RZ_{\opw}^n)$. In an FLS, the wiretapper side information is $\RZ_{\opw}^n = \RX^n \MW^{(n)}$ where $\RX$ is the base vector. Suppose the communication is of the form $\RF^{(n)} = \RX^n \MF^{(n)}$, for some matrix $\MF^{(n)}$, then the condition of  perfect alignment is equivalent to the condition that the column space of $\MF^{(n)}$ contains the column space of $\MW^{(n)}$. This is in turn equivalent to the condition that the left nullspace of $\MW^{(n)}$ contains the left nullspace of $\MF^{(n)}$, i.e., if $\Ry \MF^{(n)}=\R0$ for some vector $\Ry$ then $\Ry \MW^{(n)}=\R0$.
\end{enumerate}
So we will construct a linear communication scheme (deterministic), for some fixed $n$, achieving both perfect omniscience and perfect alignment. As a consequence,  the leakage rate for omniscience is equal to $\frac{1}{n}  I(\RZ_V^n\wedge \RF^{(n)} | \RZ_{\opw}^n) = \frac{1}{n}H( \RF^{(n)} | \RZ_{\opw}^n) = \frac{1}{n}[H( \RF^{(n)}) - H( \RZ_{\opw}^n)] = \frac{1}{n}H( \RF^{(n)}) - n_w\log_2q$. To show the desired rate, it is enough to have $\frac{1}{n}H( \RF^{(n)}) = \left(\sum_{e \in E}n_e\right) \log_2q  - \min_{e \in E} H(\RY_e) $.

This construction is given separately for multiple cases for the ease of understanding.  We start with the special case $n_e=s$ for all $e \in E$. First we consider a PIN model defined on a path graph. Then we extend it to the tree-PIN case by using the fact that there exists a unique path from any vertex to the root of the tree. Later we move to the case of arbitrary $n_e$.

\subsubsection{Path with length $L$ and $n_e=s$ for all $e \in E$}
 Let $V= \{0,1,\ldots,L\}$ be the set of vertices and $E=\{1,\ldots,L\}$ be the edge set such that edge $i$ is incident on  vertices $i-1$ and $i$. Since $n_e =s$, $\min_{e \in E} H(\RY_e)=s \log_2q$. Fix a positive integer $n$,  such that $n > \log_q(sL)$. With $n$ i.i.d. realizations of the source, the vector corresponding to edge $i$ can be expressed as $\RY_i^{n} =[ \RX^n_{i,1} \ldots \RX^n_{i,s}]$ where $\RX^n_{i,j}$'s  can be viewed as element in $\bb{F}_{q^n}$. Hence $\RY_i^{n} \in (\bb{F}_{q^n})^s$.  The goal is to construct a linear communication scheme $\RF^{(n)}$ that achieves both perfect omniscience and perfect alignment simultaneously such that $H( \RF^{(n)}) =n \left[ \left(\sum_{e \in E}n_e\right) \log_2q  - \min_{e \in E} H(\RY_e)\right] = n  \left(sL - s\right) \log_2q$.  
 
 Now we will construct the  communication as follows. Leaf nodes $0$ and $L$ do not communicate. The internal node $i$ communicates $\tRF_i^{(n)} = \RY^n_{i} + \RY^n_{i+1}\MA_{i}$, where $\MA_{i}$ is an $s \times s$ matrix with elements from $\bb{F}_{q^n}$. This  communication is of the form
\begin{align*}
\RF^{(n)} & = \begin{bmatrix}
\tRF_1^{(n)} \cdots \tRF_{L-1}^{(n)}
\end{bmatrix} \\ &= \begin{bmatrix}
\RY_1^{n}\cdots \RY_L^{n}
\end{bmatrix} \underbrace{\begin{bmatrix}
\MI & \M0& \cdots  &\M0&\M0\\
\MA_1&\MI &  \cdots &\M0&\M0\\
\M0&{\MA_2} & \cdots &\M0&\M0\\
\vdots&\vdots&\ddots&\vdots&\vdots\\
 \M0 &\M0&\cdots&\MA_{L-2}&\MI  \\
\M0 &\M0&\cdots& \M0&\MA_{L-1} \\
\end{bmatrix}}_{:=\MF^{(n)}}
\end{align*}
Here $\MF^{(n)}$ is an $sL \times s(L-1) $ matrix over $\bb{F}_{q^n}$. Observe that $\rank_{\bb{F}_{q^n}}(\MF^{(n)})= s(L-1)$, which implies that $H( \RF^{(n)}) =\left(sL - s\right) \log_2q^n$ and  the dimension of the left nullspace of $\MF^{(n)}$ is $s$. Now the communication coefficients, $(\MA_i : 1\leq i\leq L-1)$, have to be chosen such that $\RF^{(n)}$ achieves both perfect omniscience and perfect alignment. Let us derive some conditions on these matrices.

Perfect omniscience  is equivalent to the condition that the $\MA_i$'s  are invertible. The Necessity of the invertibility condition is immediate since if $\MA_{L-1}$ were not invertible, then vector $\RY_L^n$ is not completely recoverable from the communication by some users, for instance, user $0$. Sufficiency follows by observing that for any $i \in V$, $[\MF^{(n)} \mid \MH_i]$  is  full rank,  where $\MH_i$ is a block-column vector with an identity matrix at location $i$ and zero matrix in the rest of the locations. In other words, $(\RY_1^{n}\cdots \RY_L^{n})$ is recoverable from $(\RF^{(n)},  \RY_i^n)$ for any $i \in E$, hence achieving omniscience. So we assume that the $\MA_i$'s are invertible.

 For perfect alignment, we require that the left nullspace of $\MF^{(n)}$ is contained in  the left nullspace of $\MW^{(n)}$, which is the wiretapper matrix corresponding to $n$ i.i.d. realizations. Note that $\MW^{(n)}$ is a $\left(\sum_{e \in E} n_e\right) \times n_w$ matrix over $\bb{F}_{q^n}$ with entries $\MW^{(n)}(k,l) = \MW(k,l) \in  \bb{F}_{q}$; since $\bb{F}_{q} \subseteq \bb{F}_{q^n}$, $\MW^{(n)}(k,l) \in \bb{F}_{q^n}$. As pointed out before, the dimension of the  left nullspace of $\MF^{(n)}$ is $s$ whereas the dimension of the left nullspace of $\MW^{(n)}$ is $sL-n_w$. Since the source is irreducible, it follows from Lemma~\ref{lem:upbdirred} in Appendix \ref{app:nonzerodet} that $s \leq sL-n_w$. Since the dimensions are appropriate, the left nullspace inclusion condition is not impossible. Observe that 
\begin{align*}
\underbrace{\begin{bmatrix}
\MS_1 & -\MS_1\MA_1^{-1} &
\cdots&
(-1)^{L-1}\MS_1\MA_{1}^{-1}\ldots \MA_{L-1}^{-1}
\end{bmatrix}}_{:=\MS} \MF^{(n)}=\M0.
\end{align*}
where  $\MS_1$  is some invertible matrix. We write $\MS = [\MS_1 \ldots \MS_L]$ with $\MS_{i+1} :=(-1)^{i}\MS_1\MA_{1}^{-1}\ldots \MA_{i}^{-1}$ for $1 \leq i \leq L-1$. Notice that the $\MS_i$'s are invertible. We can also express the $\MA_i$'s in terms of the $\MS_i$'s as $\MA_i= -\MS_{i+1}^{-1}\MS_i$ for $1 \leq i \leq L-1$. The dimension of the left nullspace of $\MF^{(n)}$ is $s$ and all the $s$ rows of $\MS$ are independent,  so these rows span the left nullspace of $\MF^{(n)}$. Therefore for the inclusion, we must have $\MS\MW^{(n)} =\M0.$

Thus, proving the existence of communication coefficients $\MA_i$'s  that achieve perfect omniscience and perfect alignment is equivalent to proving the existence of  $\MS_i$'s  that are invertible and satisfy $[\MS_1 \ldots \MS_L]\MW^{(n)} =\M0$. To do this, we use the probabilistic method.  Consider the system of equations $[\Ry_1 \ldots \Ry_{sL}]\MW^{(n)} =\M0$ in $sL$ variables, since the matrix $\MW^{(n)}$ has full column rank, the solutions can be described in terms of  $m:=sL- n_w$ free variables. As a result, any $\MS$ that satisfies $\MS\MW^{(n)}=\M0$ can be parametrized by $ms$ variables. Without loss of generality, we assume that the submatrix of $\MS$ formed by the first $m$ columns has these independent variables, $(\Rs_{i,j}: 1\leq i \leq s, 1 \leq j \leq m)$. Knowing these entries will determine the rest of the entries of $\MS$.  So we choose $\Rs_{i,j}$'s independently and uniformly from $\bb{F}_{q^n}$.  We would like to know if there is any realization such that all the $\MS_i$'s are invertible which is equivalent to the condition $\prod_{i=1}^{L} \det(\MS_i)\neq 0$. Note that $\prod_{i=1}^{L} \det(\MS_i)$ is a multivariate polynomial in the variables, $(\Rs_{i,j}: 1\leq i \leq s, 1 \leq j \leq m)$ with degree atmost $sL$. Furthermore the polynomial is not identically zero, which follows from the irreducibility of $\MW^{(n)}$. The proof of this fact is given in  Lemma~\ref{lem:nonzeropoly} in  appendix \ref{app:nonzerodet}. Therefore, applying the  Schwartz-Zippel lemma (Lemma~\ref{lem:sz} in Appendix~\ref{app:nonzerodet}), we have
\begin{align*}
 \Pr\left\lbrace \prod_{i=1}^{L} \det(\MS_i)\neq 0\right\rbrace \geq 1- \frac{sL}{q^n} \stackrel{(a)}{>} 0 \\
\end{align*}  
where $(a)$ follows from the choice $n > \log_q(sL)$.  Since the probability is strictly positive, there exists a realization of $\MS$ such that $\MS \MW^{(n)} = 0$ and $\MS_i$'s are invertible which in turn shows the existence of a desired $\MF^{(n)}$.
\else
\subsection{Proof of Lemma~\ref{lem:irred}}
In this proof, we first identify an edge whose m.c.f. with the wiretapper's observations is a non-constant function. Then, by appropriately transforming the source, we separate out the m.c.f. from the random variables corresponding to the edge and the wiretapper. Later we argue that the source can be reduced by removing the m.c.f. component entirely without affecting  $\wskc$ and $\rl$. And we repeat this process until the source becomes irreducible. At each stage, to show that the reduction indeed leaves the m.c.f. related to the other edges  unchanged and makes the m.c.f. of the reduced edge a constant function, we  use  the following lemma which is proved in Appendix~\ref{app:nonzerodet}.
\begin{lemma}\label{lem:indgk}
 If $(\RX,\RY)$ is independent of $\RZ$, then  $\op{mcf}(\RX, (\RY,\RZ)) = \op{mcf}(\RX,\RY)$ and $\op{mcf}((\RX,\RZ), (\RY,\RZ)) = (\op{mcf}(\RX,\RY) ,\RZ)$.
\end{lemma}

 Since $(\RZ_V, \RZ_{\opw})$ is not irreducible, there exists an edge $e \in E$ such that $\RG_e := \op{mcf}(\RY_e, \RZ_{\opw})$ is a non-constant function. By using the result that the m.c.f. of a finite linear source is a linear function~\cite{chan18zero}, we can write $\RG_e =\RY_e \MM_e  =\RZ_{\opw} \MM_{\opw}$ for some full column-rank matrices, $\MM_e$ and $\MM_{\opw}$ over $\Fq$. 

We will appropriately transform the random vector $\RY_e$. Let $\MN_e$ be any matrix with full column-rank such that $\bM \MM_e \mid  \MN_e \eM$ is invertible. Define $\tRY_e := \RY_e \MN_e$, then
\begin{align*}
  \bM \RX_{e,1},\ldots,\RX_{e,n_e}\eM \bM \MM_e \mid  \MN_e \eM & = \RY_e \bM \MM_e \mid  \MN_e \eM \\
  &=\bM \RG_e, \tRY_e \eM\\
  &= \bM \RG_{e,1},\ldots,\RG_{e,\ell}, \tRX_{e,1},\ldots,\tRX_{e,\tilde{n}_e} \eM
\end{align*}
where $\tRY_e = [\tRX_{e,1},\ldots,\tRX_{e,\tilde{n}_e}]$, $\RG_e = [\RG_{e,1},\ldots,\RG_{e,\ell}]$, $\ell$ is the length of the vector $\RG_e$ and $\tilde{n}_e = n_e -\ell$. Therefore, we can obtain $(\RG_e, \tRY_e)$ by an invertible linear transformation of $\RY_e$. Note that the components $ \RG_{e,1},\ldots,\RG_{e,\ell}, \tRX_{e,1},\ldots,\tRX_{e,\tilde{n}_e}$ are also  i.i.d. random variables that are uniformly distributed over $\Fq$, and they are independent of $ \RY_{E \setminus \{e\}}:=(\RY_b: b \in E \setminus \{e\}))$. Hence $\RG_e$ is independent of $\tRY_e$ and $\RY_{E \setminus \{e\}}$.

Now we will express $\RZ_{\opw}$ in terms $\RG_e$ and $\tRY_e$.
\begin{align*}
 \RZ_{\opw} &= \RX \MW\\
 & =\RY_e\MW_e + \RY_{E \setminus \{e\}} \MW_{E \setminus \{e\}}\\
 &= \bM \RG_e & \tRY_e\eM\bM \MM_e   \MN_e \eM^{-1}\MW_e + \RY_{E \setminus \{e\}} \MW_{E \setminus \{e\}}\\
 &= \RG_e \MW^{'}_e +\tRY_e \MW^{''}_e + \RY_{E \setminus \{e\}} \MW_{E \setminus \{e\}}
\end{align*}
where the  matrices $\MW_e$ and $\MW_{E \setminus \{e\}}$ are sub-matrices  of $\MW$ formed by rows corresponding to $e$ and $E \setminus \{e\}$ respectively. Also, the matrices $\MW^{'}_e$ and $\MW^{''}_e$ are sub-matrices  of $\bM \MM_e   \MN_e \eM^{-1}\MW_e$ formed  by first $\ell$ rows and last $\tilde{n}_e$ rows respectively. Define $\tRZ_{\opw}:=\tRY_e \MW^{''}_e + \RY_{E \setminus \{e\}} \MW_{E \setminus \{e\}}$. Since $\RZ_{\opw}= \bM\RG_e & \tRZ_{\opw}\eM\bM  \MW^{'}_e \\ \MI \eM$ and $\bM\RG_e & \tRZ_{\opw}\eM = \RZ_{\opw}\bM \MM_{\opw} & \MI -\MM_{\opw}\MW^{'}_e \eM$, $\bM \RG_e & \tRZ_{\opw} \eM$ can be obtained by an invertible linear transformation of $\RZ_{\opw}$.

Since the transformations are invertible, $\RY_e$ and $\RZ_{\opw}$ can equivalently be written as $(\RG_e, \tRY_e)$ and $(\RG_e, \tRZ_{\opw} )$ respectively. We will see that $\RG_e$ can be removed from the source without affecting $\wskc$ and $\rl$.  Let us consider a  new tree-PIN  source $\tRZ_V$, which is same as $\RZ_V$ except that  $\tRY_e$ and $\tilde{n}_e$ are associated to the edge $e$, and the wiretapper side information is  $\tRZ_{\opw}$. Note that $(\tRZ_V, \tRZ_{\opw})$ is also a tree-PIN source with linear wiretapper, and $\RG_e$ is independent of $(\tRZ_V, \tRZ_{\opw})$.

 For the edge $e$, $\op{mcf}(\tRY_e, \tRZ_{\opw})$ is a  constant function. Suppose if it were a non-constant function $\tRG_e$ w.p. 1, which  is indeed independent of $\RG_e$, then $\op{mcf}(\RY_e, \RZ_{\opw}) = \op{mcf}((\RG_e, \tRY_e), (\RG_e,\tRZ_{\opw}))= (\RG_e, \tRG_e)$. The last equality uses Lemma~\ref{lem:indgk}. Therefore, $H(\RG_e) =H(\op{mcf}(\RY_e, \RZ_{\opw}))=H(\RG_e, \tRG_e) >  H(\RG_e)$, which is a contradiction.  Moreover $H(\RY_e|\op{mcf}(\RY_e, \RZ_{\opw}))= H(\RY_e|\RG_e) =H(\tRY_e, \RG_e|\RG_e) = H(\tRY_e)$. For the other edges $b \neq e$, $\tRY_b = \RY_b$ and $\op{mcf}(\tRY_b,\tRZ_{\opw})= \op{mcf}(\RY_b,\tRZ_{\opw})= \op{mcf}(\RY_b, (\RG_e,\tRZ_{\opw})) = \op{mcf}(\RY_b, \RZ_{\opw})$, which follows from Lemma~\ref{lem:indgk}.
 
Now we will verify that $\wskc$ and $\rl$ do not change.  First let us show that $\rl(\RZ_V||\RZ_{\opw}) \leq \rl(\tRZ_V||\tRZ_{\opw})$ and $\wskc(\RZ_V|| \RZ_{\opw}) \geq \wskc(\tRZ_V||\tRZ_{\opw})$.
Let $\tRF^{(n)}$  be  an optimal communication  for $\rl(\tRZ_V||\tRZ_{\opw})$. We can make use of $\tRF^{(n)}$ to construct an omniscience communication for the source $(\RZ_V,\RZ_{\opw})$. Set  $\RF^{(n)}= (\RG_e^n, \tRF^{(n)})$. This communication is made as follows. Both the terminals incident on the edge $e$  have $\RY_e^n$ or equivalently $(\RG_e^n, \tRY_e^n)$. One of them  communicates $\RG_e^n$. In addition, all the terminals communicate according to $\tRF^{(n)}$ because for every user $i$,  $\tRZ_i^n$ is recoverable from $\RZ_i^n$. It is easy to verify that this is an omniscience communication for $(\RZ_V,\RZ_{\opw})$.
The minimum rate of leakage for omniscience 
\begin{align*}
\rl(\RZ_V||\RZ_{\opw})&\leq \frac{1}{n}I(\RZ_V^n;  \RF^{(n)}|\RZ_{\opw}^n)\\ &= \frac{1}{n}I(\RZ_V^n;  \RG_e^n, \tRF^{(n)}|\RZ_{\opw}^n)\\
&\utag{a}= \frac{1}{n}I(\tRZ_V^n,\RG_e^n;  \RG_e^n, \tRF^{(n)}|\tRZ_{\opw}^n, \RG_e^n) \\ &= \frac{1}{n}I(\tRZ_V^n; \tRF^{(n)}|\tRZ_{\opw}^n, \RG_e^n) \\
&\utag{b}= \frac{1}{n}I(\tRZ_V^n; \tRF^{(n)}|\tRZ_{\opw}^n) \approx \rl(\tRZ_V||\tRZ_{\opw}),
\end{align*}
where \uref{a} is due to the fact that $(\RG_e, \tRZ_{\opw})$ is obtained by a linear invertible transformation of $\RZ_{\opw}$ and \uref{b} follows from the independence of $\RG_e$ and $(\tRZ_V, \tRZ_{\opw})$. It shows that  $\rl(\RZ_V||\RZ_{\opw}) \leq \rl(\tRZ_V||\tRZ_{\opw})$. Similarly, let $(\tRF^{(n)},\tRK^{(n)})$ be a communication and key pair  which is optimal  for  $\wskc(\tRZ_V||\tRZ_{\opw})$. By letting $(\RF^{(n)},\RK^{(n)})=( \tRF^{(n)}, \tRK^{(n)})$ for the source $(\RZ_V, \RZ_{\opw})$, we can see that the key recoverability condition is satisfied. Thus $(\RF^{(n)},\RK^{(n)})$ constitute a valid SKA scheme for $(\RZ_V, \RZ_{\opw})$ which implies that $\wskc(\RZ_V|| \RZ_{\opw}) \geq \wskc(\tRZ_V||\tRZ_{\opw})$. 

To prove the  reverse inequalities, $\rl(\RZ_V||\RZ_{\opw}) \geq \rl(\tRZ_V||\tRZ_{\opw})$ and $\wskc(\RZ_V|| \RZ_{\opw}) \leq \wskc(\tRZ_V||\tRZ_{\opw})$, we use the idea of simulating  source $(\RZ_V, \RZ_{\opw})$ from  $(\tRZ_V,\tRZ_{\opw})$. Consider the source $(\tRZ_V,\tRZ_{\opw})$ in which one of the terminals $i$ incident on the edge $e$, generates the randomness $\RG_e$ that is independent of the source and broadcasts it, after which the other terminal $j$ incident on $e$ and the wiretapper has $\RG_e$. These two terminals $i$ and $j$ simulate $\RY_e$ from $\tRY_e$ and $\RG_e$, whereas the other terminals observations are same as those of $\RZ_V$.  Hence  they can communicate according to $\RF^{(n)}$ on the simulated source $\RZ_V$.  If  $\RF^{(n)}$  achieves omniscience for $\RZ_V^n$ then so is $\tRF^{(n)}=(\RG_e^n, \RF^{(n)})$ for $\tRZ_V^n$ . Therefore the omniscience recoverability condition is satisfied. The minimum  rate of leakage for omniscience,
\begin{align*}
\rl(\tRZ_V||\tRZ_{\opw})&\leq \frac{1}{n}I(\tRZ_V^n;  \tRF^{(n)}|\tRZ_{\opw}^n)\\
&= \frac{1}{n}I(\tRZ_V^n;  \RG_e^n, \RF^{(n)}|\tRZ_{\opw}^n)\\
&= \frac{1}{n}I(\tRZ_V^n;  \RG_e^n|\tRZ_{\opw}^n)+\frac{1}{n}I(\tRZ_V^n;  \RF^{(n)}|\tRZ_{\opw}^n,\RG_e^n)\\
&\utag{a}= \frac{1}{n}I(\tRZ_V^n,\RG_e^n;  \RF^{(n)}|\tRZ_{\opw}^n,\RG_e^n) \\
&\utag{b}= \frac{1}{n}I(\RZ_V^n;  \RF^{(n)}|\RZ_{\opw}^n) \\
&\approx \rl(\RZ_V||\RZ_{\opw}),
\end{align*}
where \uref{a} follows from the independence of $\RG_e$ and $(\tRZ_V, \tRZ_{\opw})$ and \uref{b} is because $(\RG_e, \tRZ_{\opw})$ can be obtained by a linear invertible transformation of $\tRZ_{\opw}$.
This shows that $\rl(\RZ_V||\RZ_{\opw}) \geq \rl(\tRZ_V||\tRZ_{\opw})$. Similarly, if  $(\RF^{(n)}, \RK^{(n)})$ is a communication and key pair for $(\RZ_V, \RZ_{\opw})$ then terminals can communicate according to  $\tRF^{(n)}= (\RG_e^n, \RF^{(n)})$ and agree upon the key $\tRK^{(n)}= \RK^{(n)}$, which is possible due to simulation. Hence the key recoverability is immediate. The secrecy condition is also satisfied because $ I(\tRK^{(n)};  \tRF^{(n)}, \tRZ_{\opw}^n) = I(\RK^{(n)};  \RF^{(n)}, \RG_e^n, \tRZ_{\opw}^n) = I(\RK^{(n)};  \RF^{(n)}, \RZ_{\opw}^n) $. Hence $(\tRF^{(n)},\tRK^{(n)})$ forms a valid SKA scheme for $(\tRZ_V, \tRZ_{\opw})$ which implies that $\wskc(\RZ_V|| \RZ_{\opw}) \geq \wskc(\tRZ_V||\tRZ_{\opw})$.

We have shown that $\rl(\RZ_V||\RZ_{\opw}) = \rl(\tRZ_V||\tRZ_{\opw})$,  $\wskc(\RZ_V|| \RZ_{\opw}) = \wskc(\tRZ_V||\tRZ_{\opw})$  and  for the edge $e$, $\op{mcf}(\tRY_e, \tRZ_{\opw})$ is a constant function and $H(\RY_e|\op{mcf}(\RY_e, \RZ_{\opw}))= H(\tRY_e)$. Furthermore, we have shown that this  reduction does not change the m.c.f. of the $\RY_b$, which is unaffected by the reduction when $b \neq e$,  and  $\tRZ_{\opw}$, side information of the reduced wiretapper. Since  $(\tRZ_V, \tRZ_{\opw})$ is also a tree-PIN source with linear wiretapper, we can repeat this process, if it is not irreducible, until the source becomes irreducible without affecting $\wskc$ and $\rl$.

\subsection{Proof of Theorem~\ref{thm:cwsk:irred}}
\emph{Converse part.} An  upper bound on  $\wskc$  is $\skc$, because no wiretapper side information can only increase the key generation ability of users. It was shown in \cite[Example 5]{csiszar04} that if the random variables of a source form a Markov chain on a tree, then $\skc = \min_{(i,j) : \{i,j\} = \xi(e) } I(\RZ_i ; \RZ_j)$. In the tree-PIN case, which satisfies the Markov property, this turns out to be $\skc=\min_{e \in E} H(\RY_e)$. As consequence, we have $\wskc \leq \min_{e \in E} H(\RY_e)$ and 
\begin{align}\label{eq:rl:conv}
\begin{split}
 \rl&\utag{a}\geq H(\RZ_V|\RZ_{\opw}) -\wskc \\ 
 &\utag{b}= \left(\sum_{e \in E}n_e -n_w\right)\log_2q -\wskc \\
 &\geq \left(\sum_{e \in E}n_e -n_w\right)\log_2q  - \min_{e \in E} H(\RY_e)
 \end{split}
\end{align}\\
where \uref{a} follows from Proposition~\ref{thm:RL:lb} and \uref{b} is due to the full column-rank assumption on $\MW$.

\emph{Achievability part.}  In this section, we will show the existence of an omniscience scheme with leakage rate $\left(\sum_{e \in E}n_e -n_w\right)\log_2q  - \min_{e \in E} H(\RY_e)$. Hence $\rl \leq \left(\sum_{e \in E}n_e -n_w\right)\log_2q  - \min_{e \in E} H(\RY_e)$,  which together with the chain of inequalities~\eqref{eq:rl:conv} imply that $\wskc = \min_{e \in E} H(\RY_e)=\skc$  and $\rl =\left(\sum_{e \in E}n_e -n_w\right)\log_2q- \skc$. In particular, for achieving a secret key of rate $\wskc = \min_{e \in E} H(\RY_e)$, the terminals use privacy amplification on the recovered source.

In fact, the existence of an omniscience scheme is shown by first constructing a template for the communication with desired properties and then showing the existence of an instance of it by random coding. The  following are the key components involved in this construction.
\begin{enumerate}
\item \emph{Deterministic scheme:} A scheme is said to be deterministic if  terminals are  not allowed to use any locally generated private randomness. 
 \item \emph{Perfect omniscience~\cite{sirinperfect}:} For a fixed $n \in \bb{N}$, $\RF^{(n)}$ is said to achieve perfect omniscience if  terminals can recover the source $\RZ_V^n$ perfectly, i.e., $H(\RZ_V^n|\RF^{(n)}, \RZ_i^n) =0$ for all $i \in V$. If we do not allow any private randomness, then $H( \RF^{(n)} |  \RZ_V^n) = 0$, which implies
  \begin{align*}\label{eq:perfectomni}
  \begin{split}
   \frac{1}{n}  I(\RZ_V^n\wedge \RF^{(n)} | \RZ_{\opw}^n) &= \frac{1}{n}\left [H( \RF^{(n)} | \RZ_{\opw}^n) - H( \RF^{(n)} | \RZ_{\opw}^n, \RZ_V^n) \right] \\&= \frac{1}{n}H( \RF^{(n)} | \RZ_{\opw}^n).
  \end{split}
\end{align*} 
\item \emph{Perfect alignment:} For an $n \in \bb{N}$, we say that $\RF^{(n)}$ perfectly aligns with $\RZ_{\opw}^n$ if $H( \RZ_{\opw}^n|\RF^{(n)} ) = 0$. Note that $\RZ_{\opw}^n$ is only recoverable from $\RF^{(n)}$ but not the other way around. In this case, $H( \RF^{(n)} | \RZ_{\opw}^n) =H( \RF^{(n)}) - H( \RZ_{\opw}^n)$. In an FLS, the wiretapper side information is $\RZ_{\opw}^n = \RX^n \MW^{(n)}$ where $\RX$ is the base vector. Suppose the communication is of the form $\RF^{(n)} = \RX^n \MF^{(n)}$, for some matrix $\MF^{(n)}$, then the condition of  perfect alignment is equivalent to the condition that the column space of $\MF^{(n)}$ contains the column space of $\MW^{(n)}$. This is in turn equivalent to the condition that the left nullspace of $\MW^{(n)}$ contains the left nullspace of $\MF^{(n)}$, i.e., if $\Ry \MF^{(n)}=\R0$ for some vector $\Ry$ then $\Ry \MW^{(n)}=\R0$.
\end{enumerate}
So we will construct a linear communication scheme (deterministic), for some fixed $n$, achieving both perfect omniscience and perfect alignment. As a consequence,  the leakage rate for omniscience is equal to $\frac{1}{n}  I(\RZ_V^n\wedge \RF^{(n)} | \RZ_{\opw}^n) = \frac{1}{n}H( \RF^{(n)} | \RZ_{\opw}^n) = \frac{1}{n}[H( \RF^{(n)}) - H( \RZ_{\opw}^n)] = \frac{1}{n}H( \RF^{(n)}) - n_w\log_2q$. To show the desired rate, it is enough to have $\frac{1}{n}H( \RF^{(n)}) = \left(\sum_{e \in E}n_e\right) \log_2q  - \min_{e \in E} H(\RY_e) $.

This construction is given separately for multiple cases for the ease of understanding.  We start with the special case $n_e=s$ for all $e \in E$. First we consider a PIN model defined on a path graph. Then we extend it to the tree-PIN case by using the fact that there exists a unique path from any vertex to the root of the tree. Later we move to the case of arbitrary $n_e$.

\subsubsection{Path with length $L$ and $n_e=s$ for all $e \in E$}
 Let $V= \{0,1,\ldots,L\}$ be the set of vertices and $E=\{1,\ldots,L\}$ be the edge set such that edge $i$ is incident on  vertices $i-1$ and $i$. Since $n_e =s$, $\min_{e \in E} H(\RY_e)=s \log_2q$. Fix a positive integer $n$,  such that $n > \log_q(sL)$. With $n$ i.i.d. realizations of the source, the vector corresponding to edge $i$ can be expressed as $\RY_i^{n} =[ \RX^n_{i,1} \ldots \RX^n_{i,s}]$ where $\RX^n_{i,j}$'s  can be viewed as element in $\bb{F}_{q^n}$. Hence $\RY_i^{n} \in (\bb{F}_{q^n})^s$.  The goal is to construct a linear communication scheme $\RF^{(n)}$ that achieves both perfect omniscience and perfect alignment simultaneously such that $H( \RF^{(n)}) =n \left[ \left(\sum_{e \in E}n_e\right) \log_2q  - \min_{e \in E} H(\RY_e)\right] = n  \left(sL - s\right) \log_2q$.  
 
 Now we will construct the  communication as follows. Leaf nodes $0$ and $L$ do not communicate. The internal node $i$ communicates $\tRF_i^{(n)} = \RY^n_{i} + \RY^n_{i+1}\MA_{i}$, where $\MA_{i}$ is an $s \times s$ matrix with elements from $\bb{F}_{q^n}$. This  communication is of the form
\begin{align*}
\RF^{(n)} & = \begin{bmatrix}
\tRF_1^{(n)} \cdots \tRF_{L-1}^{(n)}
\end{bmatrix} \\ &= \begin{bmatrix}
\RY_1^{n}\cdots \RY_L^{n}
\end{bmatrix} \underbrace{\begin{bmatrix}
\MI & \M0& \cdots  &\M0&\M0\\
\MA_1&\MI &  \cdots &\M0&\M0\\
\M0&{\MA_2} & \cdots &\M0&\M0\\
\vdots&\vdots&\ddots&\vdots&\vdots\\
 \M0 &\M0&\cdots&\MA_{L-2}&\MI  \\
\M0 &\M0&\cdots& \M0&\MA_{L-1} \\
\end{bmatrix}}_{:=\MF^{(n)}}
\end{align*}
Here $\MF^{(n)}$ is an $sL \times s(L-1) $ matrix over $\bb{F}_{q^n}$. Observe that $\rank_{\bb{F}_{q^n}}(\MF^{(n)})= s(L-1)$, which implies that $H( \RF^{(n)}) =\left(sL - s\right) \log_2q^n$ and  the dimension of the left nullspace of $\MF^{(n)}$ is $s$. Now the communication coefficients, $(\MA_i : 1\leq i\leq L-1)$, have to be chosen such that $\RF^{(n)}$ achieves both perfect omniscience and perfect alignment. Let us derive some conditions on these matrices.

Perfect omniscience  is equivalent to the condition that the $\MA_i$'s  are invertible. The Necessity of the invertibility condition is immediate since if $\MA_{L-1}$ were not invertible, then vector $\RY_L^n$ is not completely recoverable from the communication by some users, for instance, user $0$. Sufficiency follows by observing that for any $i \in V$, $[\MF^{(n)} \mid \MH_i]$  is  full rank,  where $\MH_i$ is a block-column vector with an identity matrix at location $i$ and zero matrix in the rest of the locations. In other words, $(\RY_1^{n}\cdots \RY_L^{n})$ is recoverable from $(\RF^{(n)},  \RY_i^n)$ for any $i \in E$, hence achieving omniscience. So we assume that the $\MA_i$'s are invertible.

 For perfect alignment, we require that the left nullspace of $\MF^{(n)}$ is contained in  the left nullspace of $\MW^{(n)}$, which is the wiretapper matrix corresponding to $n$ i.i.d. realizations. Note that $\MW^{(n)}$ is a $\left(\sum_{e \in E} n_e\right) \times n_w$ matrix over $\bb{F}_{q^n}$ with entries $\MW^{(n)}(k,l) = \MW(k,l) \in  \bb{F}_{q}$; since $\bb{F}_{q} \subseteq \bb{F}_{q^n}$, $\MW^{(n)}(k,l) \in \bb{F}_{q^n}$. As pointed out before, the dimension of the  left nullspace of $\MF^{(n)}$ is $s$ whereas the dimension of the left nullspace of $\MW^{(n)}$ is $sL-n_w$. Since the source is irreducible, it follows from Lemma~\ref{lem:upbdirred} in Appendix \ref{app:nonzerodet} that $s \leq sL-n_w$. Since the dimensions are appropriate, the left nullspace inclusion condition is not impossible. Observe that 
\begin{align*}
\underbrace{\begin{bmatrix}
\MS_1 & -\MS_1\MA_1^{-1} &
\cdots&
(-1)^{L-1}\MS_1\MA_{1}^{-1}\ldots \MA_{L-1}^{-1}
\end{bmatrix}}_{:=\MS} \MF^{(n)}=\M0.
\end{align*}
where  $\MS_1$  is some invertible matrix. We write $\MS = [\MS_1 \ldots \MS_L]$ with $\MS_{i+1} :=(-1)^{i}\MS_1\MA_{1}^{-1}\ldots \MA_{i}^{-1}$ for $1 \leq i \leq L-1$. Notice that the $\MS_i$'s are invertible. We can also express the $\MA_i$'s in terms of the $\MS_i$'s as $\MA_i= -\MS_{i+1}^{-1}\MS_i$ for $1 \leq i \leq L-1$. The dimension of the left nullspace of $\MF^{(n)}$ is $s$ and all the $s$ rows of $\MS$ are independent,  so these rows span the left nullspace of $\MF^{(n)}$. Therefore for the inclusion, we must have $\MS\MW^{(n)} =\M0.$

Thus, proving the existence of communication coefficients $\MA_i$'s  that achieve perfect omniscience and perfect alignment is equivalent to proving the existence of  $\MS_i$'s  that are invertible and satisfy $[\MS_1 \ldots \MS_L]\MW^{(n)} =\M0$. To do this, we use the probabilistic method.  Consider the system of equations $[\Ry_1 \ldots \Ry_{sL}]\MW^{(n)} =\M0$ in $sL$ variables, since the matrix $\MW^{(n)}$ has full column rank, the solutions can be described in terms of  $m:=sL- n_w$ free variables. As a result, any $\MS$ that satisfies $\MS\MW^{(n)}=\M0$ can be parametrized by $ms$ variables. Without loss of generality, we assume that the submatrix of $\MS$ formed by the first $m$ columns has these independent variables, $(\Rs_{i,j}: 1\leq i \leq s, 1 \leq j \leq m)$. Knowing these entries will determine the rest of the entries of $\MS$.  So we choose $\Rs_{i,j}$'s independently and uniformly from $\bb{F}_{q^n}$.  We would like to know if there is any realization such that all the $\MS_i$'s are invertible which is equivalent to the condition $\prod_{i=1}^{L} \det(\MS_i)\neq 0$. Note that $\prod_{i=1}^{L} \det(\MS_i)$ is a multivariate polynomial in the variables, $(\Rs_{i,j}: 1\leq i \leq s, 1 \leq j \leq m)$ with degree atmost $sL$. Furthermore the polynomial is not identically zero, which follows from the irreducibility of $\MW^{(n)}$. The proof of this fact is given in  Lemma~\ref{lem:nonzeropoly} in  appendix \ref{app:nonzerodet}. Therefore, applying the  Schwartz-Zippel lemma (Lemma~\ref{lem:sz} in Appendix~\ref{app:nonzerodet}), we have
\begin{align*}
 \Pr\left\lbrace \prod_{i=1}^{L} \det(\MS_i)\neq 0\right\rbrace \geq 1- \frac{sL}{q^n} \stackrel{(a)}{>} 0 \\
\end{align*}  
where $(a)$ follows from the choice $n > \log_q(sL)$.  Since the probability is strictly positive, there exists a realization of $\MS$ such that $\MS \MW^{(n)} = 0$ and $\MS_i$'s are invertible which in turn shows the existence of a desired $\MF^{(n)}$.

\subsubsection{Tree with $L$ edges and $n_e=s$ for all $e \in E$}
For tree-PIN model, we essentially use the same kind of communication construction as that of the path model.  Consider a PIN model on a tree with $L+1$ nodes and $L$ edges. To describe the linear communication, fix some leaf node as the root, $\rho$, of the tree. For any internal node $i$ of the tree, let $E_i$ denote the edges incident with $i$, and in particular, let $e^*(i)\in E_i$ denote the edge incident with $i$ that is on the unique path between $i$ and $\rho$. Fix a positive integer $n$,  such that $n > \log_q(sL)$. The communication from an internal node $i$ is  $(  \RY^n_{e^*(i)} + \RY^n_{e}\MA_{i,e}: e \in E_i \setminus \{e^*(i)\})$, where $\MA_{i,e}$ is an $s \times s$ matrix.  Each internal node communicates $s(d_i - 1)$ symbols from $\bb{F}_{q^n}$, where $d_i$ is the degree of the node $i$. Leaf nodes do not communicate. The total number of $\bb{F}_{q^n}$-symbols communicated is $\sum_i s(d_i-1)$, where the sum is over all nodes, including leaf nodes. The contribution to the sum from leaf nodes is in fact $0$, but including all nodes in the sum allows us to evaluate the sum as $s[2 \times(\text{number of edges}) - (\text{number of nodes})] = s(L-1)$.
 Thus, we have the overall communication of the form 
\begin{align*}
 \RF^{(n)} = \RY^n \MF^{(n)} 
\end{align*}
 where $\MF^{(n)} $ is a $sL \times s(L-1)$ matrix over $\bb{F}_{q^n}$ and $\RY^n = (\RY^n_e)$. The rows of $\MF^{(n)}$  correspond to the edges  of the tree.  The aim is to  choose the matrices $\MA_i$ that achieves both perfect omniscience and perfect alignment simultaneously such that $H( \RF^{(n)}) =n \left[ \left(\sum_{e \in E}n_e\right) \log_2q  - \min_{e \in E} H(\RY_e)\right] = n  \left(sL - s\right) \log_2q$.  
 
For perfect omniscience,  it is sufficient for the $\MA_i$'s to be     invertible. First observe that all the leaf nodes are connected to  the root node $\rho$  via paths. On each of these paths the communication has exactly the same form as that of the path model considered before. So when the $\MA_i$'s are invertible, the root node can recover the entire source using  $\RY_{e_\rho}^n$, where $e_\rho$ is the edge incident on $\rho$. Now take any node $i$, there is a unique path from $i$ to $\rho$. Again the form of the communication restricted to this path is same as that of the path model. Hence node $i$,  just using $\RY_{e^*(i)}^n$ can  recover $\RY_{e_\rho}^n$ , which in turn, along with the overall communication, allows node $i$ to recover the entire source. Indeed, only edge observations $\RY_e^n$ are used in the recovery process.  

 Because $\RY^n$ is recoverable from $(\RF^{(n)},  \RY_e^n)$ for any $e \in E$, $[\MF^{(n)} \mid \MH_e]$ is an invertible $sL \times sL$ matrix, where $\MH_e$ is a block-column vector with an $s \times s$ identity matrix at location corresponding to edge $e$ and zero matrix in the rest of the locations. Therefore $\MF^{(n)}$ is a full column-rank matrix, i.e., $\rank_{\bb{F}_{q^n}}(\MF^{(n)})= s(L-1)$, which implies that $H( \RF^{(n)}) =\left(sL - s\right) \log_2q^n$ and  the dimension of the left nullspace of $\MF^{(n)}$ is $s$.

For perfect alignment, we require that the left nullspace of $\MF^{(n)}$ is contained in  the left nullspace of $\MW^{(n)}$. So, let us construct an $\MS = (\MS_e)$ such that  $ \MS\MF^{(n)}=\M0$ as follows. Let $\MS_1$ be an invertible matrix. Each edge $e$ has two nodes incident with it; let $i^*(e)$ denote the node that is closer to the root $\rho$. There is a unique path $i^*(e) = i_1 \longrightarrow i_2 \longrightarrow  \cdots \longrightarrow i_{\ell} = \rho$ that connects $i^*(e)$ to $\rho$ and let the edges along the path in this order is $(e=e_1, e_2,\ldots, e_{\ell})$. We set $\MS_e := (-1)^{\ell-1} \MS_1 \MA^{-1}_{i_{\ell -1}, e_{\ell-1}} \ldots \MA^{-1}_{i_{1}, e_{1}} $ for all edges $e$ except for the edge incident with $\rho$, to which we associate $S_1$. Note that the $\MS_e$'s are invertible and $\MS_e= - \MS_{e^{\#}}\MA^{-1}_{i^*(e), e}$, where  $e^{\#}$ is the edge adjacent to $e$ on the unique path from $i^*(e)$ to $\rho$. Let us now verify that $\MS \MF^{(n)} = \M0$. The  component corresponding to the internal node $i$ in $\MS \MF^{(n)}$ is of the form $(\MS_{e^*(i)} + \MS_{e}\MA_{i,e}: e \in E_i \setminus \{e^*(i)\})$. But for an  $e \in E_i \setminus \{e^*(i)\}$, $i^{*}(e) = i$ and $e^{\#} = e^*(i)$, thus $\MS_{e}\MA_{i,e} = - \MS_{e^{\#}}\MA^{-1}_{i^*(e), e}\MA_{i,e}= - \MS_{e^*(i)}\MA^{-1}_{i, e}\MA_{i,e} =- \MS_{e^*(i)}$. Hence we have  $\MS_{e^*(i)} + \MS_{e}\MA_{i,e}=\M0$ which implies $\MS \MF^{(n)} = \M0$.
The dimension of the left nullspace of $\MF^{(n)}$ is $s$ and all the $s$ rows of $\MS$ are independent,  so these rows span the left nullspace of $\MF^{(n)}$. Therefore for the inclusion, we must have $\MS\MW^{(n)} =\M0$.

Finally, we can prove the existence of $\MS$ such that $\MS\MW^{(n)} =\M0$ and $\MS_i$'s are invertible, using the probabilistic method exactly as before. The details are omitted.  This shows the existence of a desired $\MF^{(n)}$.

\subsubsection{Path with length $L$  and arbitrary $n_e$}
Define $s := \min\{n_e: e \in E\}$. In this case, the communication consists of two parts. One part involves the communication that is similar to that of the $n_e =s$ case, where we use first $s$ random variables associated to each edge $e$. The other part involves revealing the rest of the random variables on each edge, but this is done by linearly combining them with the first $s$ rvs. 

Let $V= \{0,1,\ldots,L\}$ be the set of vertices and $E=\{1,\ldots,L\}$ be the edge set such that edge $i$ is incident on  vertices $i-1$ and $i$. Fix a positive integer $n$,  such that $n > \log_q(sL)$.  As before, with $n$ i.i.d. realizations of the source, the vector corresponding to edge $i$ can be expressed as $\RY_i^{n} =[ \RX^n_{i,1} \ldots \RX^n_{i,s} \RX^n_{i,s+1}\ldots \RX^n_{i,n_i}]$ where $\RX^n_{i,j}$'s are viewed as element in $\bb{F}_{q^n}$. Hence $\RY_i^{n} \in (\bb{F}_{q^n})^{n_i}$. Since $s = \min\{n_e: e \in E\}$, we have $\min_{e \in E} H(\RY_e)=s \log_2q$.  The goal is again to construct a linear communication scheme $\RF^{(n)}$ that achieves both perfect omniscience and perfect alignment simultaneously such that $H( \RF^{(n)}) =n \left[ \left(\sum_{e \in E}n_e\right) \log_2q  - \min_{e \in E} H(\RY_e)\right] = n  \left( \sum_{e \in E}n_e  -s \right)\log_2q$. 

Now we will construct the  communication as follows. The leaf node $0$  does not communicate. The internal node $i$ communicates $\tRF_i^{(n)} = \RY^n_{i}\bM \MI &  \MB_{i}\\ \M0 & \MI \eM + \RY^n_{i+1}\bM \MA_{i} & \M0 \\ \M0 & \M0 \eM$, where $\MA_{i}$ is an $s \times s$ matrix and $\MB_{i}$ is an $s \times (n_i -s)$ matrix with elements from $\bb{F}_{q^n}$.  The communication from the leaf node $L$ is $\tRF_{L}^{(n)} = \RY^n_{L}\bM \MB_{L} \\ \MI \eM$, where $\MB_{L-1}$ is an $s \times (n_L -s)$ matrix. This  communication is of the form $\RF^{(n)} = \bM
\tRF_1^{(n)}& \cdots& \tRF_{L-1}^{(n)} & \tRF_{L}^{(n)} \eM = \bM \RY_1^{n}\cdots \RY_L^{n} \eM \MF^{(n)} $ where $\MF^{(n)}$ is
\begin{align*}
 \renewcommand{\arraystretch}{1.25}
\left[\begin{array}{c|c|c|c|c}
 \begin{matrix}
  \MI& \MB_1 \\
  \M0 & \MI
  \end{matrix}
  & \M0 & \cdots & \M0 & \M0 \\ \hline
   \begin{matrix}
  \MA_1& \M0 \\
  \M0 & \M0
  \end{matrix} &
  \begin{matrix}
  \MI & \MB_2 \\
  \M0 & \MI
  \end{matrix}
   & \cdots & \M0 & \M0 \\ \hline
  \vdots&\vdots&\ddots&\vdots&\vdots\\ \hline
  \M0 &\M0&\cdots& \begin{matrix}
  \MI & \MB_{L-1} \\
  \M0 & \MI
  \end{matrix}
  &\M0 \\ \hline
\M0 &\M0&\cdots&  \begin{matrix}
  \MA_{L-1}& \M0 \\
  \M0 & \M0
  \end{matrix}& \begin{matrix}
  \MB_{L} \\ \MI
  \end{matrix}
\end{array}\right]
\end{align*}
which is a $\left(\sum_{e \in E}n_e\right) \times  \left(\sum_{e \in E}n_e -s\right) $ matrix over $\bb{F}_{q^n}$. Observe that  $\rank_{\bb{F}_{q^n}}(\MF^{(n)})= \left(\sum_{e \in E}n_e -s\right)$, which implies that $H( \RF^{(n)}) =\left(\sum_{e \in E}n_e -s\right) \log_2q^n$ and  the dimension of the left nullspace of $\MF^{(n)}$ is $s$. Now the communication coefficients, $(\MA_i : 1\leq i\leq L-1)$ and $(B_i : 1\leq i\leq L)$  have to be chosen such that $\RF^{(n)}$ achieves both perfect omniscience and perfect alignment. As before, we derive some conditions on these matrices.

For perfect omniscience, invertibility of $\MA_i$'s  is sufficient with  no additional assumption on $\MB_i$'s. This follows by observing that when all the $\MA_i$'s are invertible then  for any $i \in V$, $[\MF^{(n)} \mid \MH_i]$  is full rank,  where $\MH_i$ is a block-column vector with $\bM \MI & \M0 \eM^T$  at location $i$ and zero matrix in the rest of the locations. In other words, $(\RY_1^{n}\cdots \RY_L^{n})$ is recoverable from $(\RF^{(n)}, (\RX^n_{i,1} \ldots \RX^n_{i,s}))$ for any $i \in E$, which means that the first $s$ random variables of each edge are enough to achieve omniscience. So we assume that the $\MA_i$'s are invertible with no restriction on the $\MB_i$'s.

For perfect alignment, we require that the left nullspace of $\MF^{(n)}$ is contained in the left nullspace of $\MW^{(n)}$. which is the wiretapper matrix corresponding to $n$ i.i.d. realizations. As pointed out earlier, the dimension of the  left nullspace of $\MF^{(n)}$ is $s$ whereas the dimension of the left nullspace of $\MW^{(n)}$ is $\left(\sum_{e \in E}n_e\right)-n_w$. Since the source is irreducible, it follows from Lemma~\ref{lem:upbdirred} in appendix \ref{app:nonzerodet} that $s \leq \left(\sum_{e \in E}n_e\right)-n_w$. Since the dimensions are appropriate, the left nullspace inclusion condition is not impossible. Observe that 
\begin{align*}
\underbrace{ \bM \MS_1 & \MT_1 \mid & \ldots& \mid \MS_L& \MT_L \eM}_{:=\MS} \MF^{(n)}=\M0.
\end{align*}
where  $\MS_1$  is some invertible matrix , $\MS_{i+1} :=(-1)^{i}\MS_1\MA_{1}^{-1}\ldots \MA_{i}^{-1}$ for $1 \leq i \leq L-1$ , $\MT_1 =-\MS_1\MB_{1}$ and $\MT_{i} = (-1)^{i}\MS_1\MA_{1}^{-1}\ldots \MA_{i-1}^{-1}\MB_{i}$ for $2 \leq i \leq L $. Notice that $\MS_i$'s are invertible. We can also express the $\MA_i$'s in terms of $\MS_i$'s as $\MA_i= -\MS_{i+1}^{-1}\MS_i$ for $1 \leq i \leq L-1$, and $\MB_i$'s in terms of $\MS_i$'s and $\MT_i$'s as $\MB_i= -\MS_{i}^{-1}\MT_i$ for $1 \leq i \leq L$. The dimension of the left nullspace of $\MF^{(n)}$ is $s$ and all the $s$ rows of $\MS$ are independent,  so these rows span the left nullspace of $\MF^{(n)}$. Therefore for the inclusion, we must have $\MS\MW^{(n)} =\M0.$

Thus, proving the existence of the communication coefficients $\MA_i$'s and $\MB_i$'s  that achieve perfect omniscience and perfect alignment is equivalent to proving the existence of $\bM \MS_1 & \MT_1 \mid & \ldots& \mid \MS_L& \MT_L \eM$ satisfying $\bM \MS_1 & \MT_1 \mid & \ldots& \mid \MS_L& \MT_L \eM \MW^{(n)} =\M0$ such that the $\MS_i$'s  are invertible. To do this, we use the probabilistic method.  Consider the system of equations $[\Ry_1 \ldots \Ry_{(\sum_{e \in E} n_e)}]\MW^{(n)} =\M0$ in $\left(\sum_{e \in E} n_e\right)$ variables, since the matrix $\MW^{(n)}$  has full column rank, the solutions can be described in terms of  $m:=\left(\sum_{e \in E} n_e\right)- n_w$ free variables. As a result, any $\MS$ that satisfies $\MS\MW^{(n)}=\M0$ can be parametrized by $ms$ variables. Without loss of generality, we assume that the submatrix of $\MS$ formed by the first $m$ columns has these independent variables, $(\Rs_{i,j}: 1\leq i \leq s, 1 \leq j \leq m)$. Knowing these entries will determine the rest of the entries of $\MS$.  So we choose $\Rs_{i,j}$'s independently and uniformly form $\bb{F}_{q^n}$.  We would like to know if there is any realization such that all  $\MS_i$'s are invertible which is equivalent to the condition $\prod_{i=1}^{L} \det(\MS_i)\neq 0$. Note that $\prod_{i=1}^{L} \det(\MS_i)$ is a multivariate polynomial in the variables, $(\Rs_{i,j}: 1\leq i \leq s, 1 \leq j \leq m)$ with degree atmost $sL$. Furthermore the polynomial is not identically zero, which follows from the irreducibility of $\MW^{(n)}$. The proof of this fact is given in  Lemma~\ref{lem:nonzeropoly} in  Appendix~\ref{app:nonzerodet}. Therefore, applying the  Schwartz-Zippel lemma (Lemma~\ref{lem:sz} in Appendix~\ref{app:nonzerodet}), we have
\begin{align*}
 \Pr\left\lbrace \prod_{i=1}^{L} \det(\MS_i)\neq 0\right\rbrace \geq 1- \frac{sL}{q^n} \stackrel{(a)}{>} 0 \\
\end{align*}  
where $(a)$ follows from the choice $n > \log_q(sL)$.  Since the probability is strictly positive, there exists a realization of $\MS$ such that $\MS \MW^{(n)} = 0$ and $\MS_i$'s are invertible which in turn shows the existence of a desired $\MF^{(n)}$.

\subsubsection{Tree with $L$ edges and arbitrary $n_e$}
For this general most case, we construct a communication scheme similar to that of the general path model by making use of the idea that there is a unique path from a node to the root of the tree. Define $s := \min\{n_e: e \in E\}$.   Consider a PIN model on a tree with $L+1$ nodes and $L$ edges. To describe the linear communication, fix some leaf node as the root, $\rho$, of the tree. For any internal node $i$ of the tree, let $E_i$ denote the edges incident with $i$, and in particular, let $e^*(i)\in E_i$ denote the edge incident with $i$ that is on the unique path between $i$ and $\rho$. Fix a positive integer $n$,  such that $n > \log_q(sL)$. We split  $\RY_e^{n} =\left[ \RX^n_{e,1} \ldots \RX^n_{e,s} \RX^n_{e,s+1}\ldots \RX^n_{e,n_e} \right] $ into two parts namely $\RY_{e[s]}^{n} =\left[ \RX^n_{e,1} \ldots \RX^n_{e,s} \right] $ and $\RY_{e[s+1,n_e]}^{n} =\left[ \RX^n_{e,1} \ldots \RX^n_{e,s} \right] $. The communication involves two parts. First part consists of communication involving $\RY_{e[s]}^{n}$. The communication from an internal node $i$ is the tuple $(  \RY^n_{e^*(i)[s]} + \RY^n_{e[s]}\MA_{i,e}: e \in E_i \setminus e^*(i))$, where $\MA_{i,e}$ is an $s \times s$ matrix. Leaf nodes do not communicate. This communication is exactly the same as that of the general path model except that it uses only $s$ random variables. The second part involves the remaining  random variables $\RY_{e[s+1,n_e]}^{n}$. Except the root node, all the other nodes communicate as follows: node $i \neq \rho$ communicates $( \RY^n_{e^*(i)[s]} \MB_{e^*(i)}+ \RY^n_{e^*(i)[s+1, n_{e^*(i)}]})$ where $ \MB_{e^*(i)}$ is an $ (n_e -s) \times s$ matrix. Number of $\Fq^n$-symbols  communicated is $ s(L-1)+\sum_{e \in E}(n_e -s)= \sum_{e \in E}n_e -s$. Thus, we have the overall communication of the form 
\begin{align*}
 \RF^{(n)} = \RY^n \MF^{(n)} 
\end{align*}
 where $\MF^{(n)} $ is a $(\sum_{e \in E}n_e) \times \left(\sum_{e \in E}n_e-s\right)$ matrix over $\bb{F}_{q^n}$ and $\RY^n = (\RY^n_e)$. The rows of $\MF^{(n)}$  correspond to the edges  of the tree.  The aim is to  choose the matrices $\MA_i$ and $\MB_i$ that achieve both perfect omniscience and perfect alignment simultaneously such that $H( \RF^{(n)}) =n \left[ \left(\sum_{e \in E}n_e\right) \log_2q  - \min_{e \in E} H(\RY_e)\right] = n  \left[ \left(\sum_{e \in E}n_e\right)  - s\right] \log_2q$.  
  
For perfect omniscience,  it is sufficient for  $\MA_i$'s to be     invertible. If $\MA_i$'s are  invertible, then as in the tree-PIN case with constant $n_e$, the nodes can recover $\RY_{[s]}^{n}$ using the first part of the communication. The partially recover source, $\RY_{[s]}^{n}$ together with the second part of the communication allows the nodes to recover the entire source $\RY^n$.  In fact, while decoding  node $i$  just uses $\RY_{e^*(i)[s]}^n$ to  attain omniscience. In other words, $\RY^n$ is recoverable from $(\RF^{(n)},  \RY_{e[s]}^n)$ for any $e \in E$. Hence $[\MF^{(n)} \mid \MH_e]$ is an invertible $\left(\sum_{e \in E}n_e\right) \times \left(\sum_{e \in E}n_e\right)$ matrix where $\MH_e$ is a block-column vector with $\bM \MI \\ \M0 \eM_{n_e \times s}$  at location corresponding to edge $e$ and zero matrix in the rest of the locations.  This shows that $\MF^{(n)}$ is a full column-rank matrix, i.e., $\rank_{\bb{F}_{q^n}}(\MF^{(n)})= \left(\sum_{e \in E}n_e\right) -s$, which implies that $H( \RF^{(n)}) =n  \left[ \left(\sum_{e \in E}n_e\right)  - s\right]  \log_2q^n$ and  the dimension of the left nullspace of $\MF^{(n)}$ is $s$.

For perfect alignment, we require that the left nullspace of $\MF^{(n)}$ is contained in  the left nullspace of $\MW^{(n)}$. So, let us construct an $\MS = (\MS_e, \MT_e)$ , where $\MS_e$ is an $s \times s$ matrix and $\MT_e$ is an $s \times (n_e-s)$ matrix such that  $ \MS\MF^{(n)}=\M0$ as follows. Let $\MS_1$ be an invertible matrix. Each edge $e$ has two nodes incident with it; let $i^*(e)$ denote the node that is closer to the root $\rho$. There is a unique path $i^*(e) = i_1 \longrightarrow i_2 \longrightarrow  \cdots \longrightarrow i_{\ell} = \rho$ that connects $i^*(e)$ to $\rho$ and let the edges along the path in this order is $(e=e_1, e_2,\ldots, e_{\ell})$. Denote the edge incident on $\rho$ by $e (\rho)$, we set $\MS_{e(\rho)}=\MS_1$, $\MS_e := (-1)^{\ell-1} \MS_1 \MA^{-1}_{i_{\ell -1}, e_{\ell-1}} \ldots \MA^{-1}_{i_{1}, e_{1}} $ for $e \neq e(\rho)$,  $\MT_{e(\rho)}=- \MS_{e(\rho)}\MB_{e(\rho)}$ and  $\MT_e := -\MS_e\MB_e$ for $e \neq e(\rho)$. Note that $\MS_e$'s are invertible and $\MS_e= - \MS_{e^{\#}}\MA^{-1}_{i^*(e), e}$, where  $e^{\#}$ is the edge adjacent to $e$ on the unique path from $i^*(e)$ to $\rho$. Let us now verify that $\MS \MF^{(n)} = \M0$. The  component corresponding to the internal node $i$ from first part of communication in $\MS \MF^{(n)}$ is of the form $(\MS_{e^*(i)} + \MS_{e}\MA_{i,e}: e \in E_i \setminus \{e^*(i)\})$. But for an  $e \in E_i \setminus \{e^*(i)\}$, $i^{*}(e) = i$ and $e^{\#} = e^*(i)$, thus $\MS_{e}\MA_{i,e} = - \MS_{e^{\#}}\MA^{-1}_{i^*(e), e}\MA_{i,e}= - \MS_{e^*(i)}\MA^{-1}_{i, e}\MA_{i,e} =- \MS_{e^*(i)}$. Hence we have  $\MS_{e^*(i)} + \MS_{e}\MA_{i,e}=\M0$. The  component corresponding to the  node $i \neq \rho$ from second part of communication in $\MS \MF^{(n)}$ is of the form $(\MS_{e^*(i)} \MB_{e^*(i)}+ \MT_{e^*(i)})$, which is $\M0$ from the choice of $\MT_{e^*(i)}$. This shows that  $\MS \MF^{(n)} = \M0$. Moreover, the dimension of the left nullspace of $\MF^{(n)}$ is $s$ and all the $s$ rows of $\MS$ are independent,  so these rows span the left nullspace of $\MF^{(n)}$. Therefore for the inclusion, we must have $\MS\MW^{(n)} =\M0$.

Finally, we can prove the existence of $\MS= (\MS_e, \MT_e)$ such that $\MS\MW^{(n)} =\M0$ and $\MS_i$'s are invertible, using the probabilistic method exactly as in the general path model. The details are omitted.  This shows the existence of a desired $\MF^{(n)}$.

\fi

\ifPAGELIMIT
\else
\section{Explicit $\rl$ Protocol in the case $n_e=1$ for all $e \in E$}\label{sec:prot}
In the proof of Theorem \ref{thm:cwsk:irred}, we have fixed the communication matrix structure and argued, using the probabilistic method, that if $n > \log_q sL$ then there exist communication coefficients  that achieve $\rl$. To that end, we first showed the existence of a realization of an $\MS$ such that $\MS\MW^{(n)}=0$ and $\MS_i$'s are invertible. Since $\MS$ and communication coefficients are recoverable from each other, the desired existence follows. However, in the case when $n_e=1$ for all $e \in E$, we give an explicit way to find these coefficients and the sufficient $n$ to do this. Here $\MS$ is just a row vector with entries from $\bb{F}_{q^n}$. Our goal is to find a vector with non-zero entries from $\bb{F}_{q^n}$  for some $n$ such that it satisfies $\MS\MW^{(n)} =\M0$.  Note that $\MW^{(n)}$ is a $\left(\sum_{e \in E} n_e\right) \times n_w$ matrix over $\bb{F}_{q^n}$ with entries $\MW^{(n)}(k,l) = \MW(k,l) \in  \bb{F}_{q}$; since $\bb{F}_{q} \subseteq \bb{F}_{q^n}$, $\MW^{(n)}(k,l) \in \bb{F}_{q^n}$. In the proof of the  following lemma, we actually show how to choose $\MS$.
 \begin{lemma}
  Let $\MW$ be an $(m+k) \times m$ matrix over $\Fq$ and $k,m \geq 1$. Assume that the columns of $\MW$ are linearly independent. If the span of the columns of $\MW$ does not contain any vector that is a scalar multiple of any standard basis vector, then there exists an $1 \times  (m+k)$ vector $\MS$ whose entries belong to $\mathbb{F}_{q^k}^{\times} := \mathbb{F}_{q^k} \backslash \{0\}$ such that $\MS\MW^{(k)}=0$.
 \end{lemma}
 \begin{proof}
  Since the columns of $\MW$ are linearly independent, we can apply elementary column operations and row swappings on the matrix $\MW$ to reduce into the form $\tMW=[\MI_{m\times m} \mid \MA_{m \times k}]^T$, for some matrix $\MA_{m \times k}$. It means that $\tMW=\MP\MW\MC$ for some permutation matrix $\MP$ and an invetible matrix $\MC$ corresponding to the column operations. Furthermore, the matrix $\MA_{m \times k}$ has no zero rows because if there were a zero row in $\MA$ then the corresponding column of $\tMW$ is a standard basis vector which means that the columns of $\MW$ span a standard basis vector contradicting the hypothesis.
  
  Now consider the field $\mathbb{F}_{q^k}$. The condition $\MS\MW^{(k)} =0$ can be written as  $\tMS\tMW^{(k)} = 0$ where $\tMW^{(k)}=\MP\MW^{(k)}\MC$ and $\tMS=\MS\MP^{-1}$. Since $\mathbb{F}_{q^k}$ is a vector space over $\Fq$, there exists a basis  $\{\beta_1, \beta_2, \ldots, \beta_k\} \subset \mathbb{F}_{q^k}$. We will use this basis to construct $\tMS$ and hence $\MS$. For $\MA=[a_{ij}]_{i \in [m], j \in [k]}$, set $\tMS_{m+i} = \beta_i\neq 0$ for $i \in [k]$ and $\tMS_i = - \sum_{j=1} ^{k} a_{ij}\beta_j \neq 0$ for $i \in [m]$. So all entries of $\tMS$  are non-zero entries which follows from the fact that $\beta_j$'s are linearly independent and for a fixed $i$, $a_{ij}$'s are not all zero. Therefore we found an $\tMS$ such that $\tMS\tMW^{(k)} = \M0$. This in turns gives $\MS$, which is obtained by permuting the columns of $\tMS$, such that $\MS\MW^{(k)} = \M0$.
 \end{proof}
 
 In the case when $n_e=1$ and the source is irreducible, the wiretapper matrix satisfies the conditions in the hypothesis of the above lemma. Therefore, we can use the construction given in that lemma to find an $\MS$  such that $\MS\MW^{(n)}=0$ where $n = |E|-n_w$. From $\MS$, we can recover back the communication coefficients $\MA_{i,e} \in \mathbb{F}_{q^k}$ because  given all $\MS_e$ along the unique path from $i$ to the root node, we can recursively compute all $\MA_{i,e}$ along that path. 
 
 We could not extend these ideas beyond this case but it is worth finding such simple and explicit constructions in the arbitrary $n_e$ case. Another interesting question is, for a given tree-PIN source with linear wiretapper, what is the minimum $n$ required to achieve perfect omniscience and perfect alignment using a linear communication? Note that the $n$ required in our protocol is $|E|-n_w$ whereas the  probabilistic method guarantees a scheme if $n > \log_q|E|$. So we clearly see that  $n = |E|-n_w$  is not optimal in some cases. 
 
\fi

\section{Conclusion and Future direction}\label{sec:conc}
For a tree-PIN model with linear wiretapper, we have characterized minimum leakage rate for omniscience and wiretap secret key capacity. Also we showed that a linear and non-interactive scheme achieves these quantities. Moreover we constructed an explicit protocol that is optimal in the case of $n_e=1$ for all $e \in E$, but we resorted to random coding approach for the general case. It is of interest to have a deterministic coding scheme covering the general case, which is left open.  We conjecture that, for  finite linear sources, the $\wskc$ can also be obtained through secure omniscience, and a linear protocol is sufficient.  However, proving this even for a general PIN model turned out to be quite challenging.

\ifPAGELIMIT
\newpage
\fi 

\bibliographystyle{IEEEtran}
\bibliography{IEEEabrv,ref}

\ifPAGELIMIT
\else
\appendix
\section{Appendix} \label{app:nonzerodet}
\subsection{Proof of Lemma~\ref{lem:indgk}}
 Any common function (c.f.) of $\RX$ and $\RY$ is also a common function of $\RX$ and $(\RY,\RZ)$. Let $\RF$ be a c.f. of $\RX$ and $(\RY,\RZ)$ which means that $H(\RF|\RX)=0=H(\RF|\RY,\RZ)$. Note that $H(\RF|\RY)=H(\RZ|\RY)+H(\RF|\RZ,\RY)-H(\RZ|\RF,\RY)=H(\RZ)-H(\RZ|\RF,\RY)$, where the last equality uses independence of $\RZ$ and $(\RX,\RY)$. Also we have $H(\RZ|\RF,\RY) \geq H(\RZ|\RX,\RY)$ which follows from the  fact that $\RF$ is a function of $\RX$. Both these inequalities together imply that $0 \le H(\RF|\RY) \leq H(\RZ)-H(\RZ|\RX,\RY) =0$. So any c.f. of $\RX$ and $(\RY,\RZ)$ is also a c.f. of $\RX$ and $\RY$.  Therefore $\op{mcf}(\RX, (\RY,\RZ)) = \op{mcf}(\RX,\RY)$. 
 
 We can see that $ (\op{mcf}(\RX,\RY) ,\RZ)$ is a c.f. of $(\RX,\RZ)$ and $(\RY,\RZ))$. To show that $\op{mcf}((\RX,\RZ), (\RY,\RZ)) = (\op{mcf}(\RX,\RY), \RZ)$, it is enough to show that $H(\op{mcf}(\RX,\RY) ,\RZ) \geq H(\RG)$ for any $\RG$ satisfying $H(\RG|\RX,\RZ)=0=H(\RG|\RY,\RZ)$. Since $\sum_{\Rz \in \mc{Z}}P_{\RZ}(\Rz) H(\RG|\RX,\RZ=\Rz)=H(\RG|\RX,\RZ)=0$, for a $\Rz \in \op{supp}(P_{\RZ})$,  we have $H(\RG|\RX,\RZ=\Rz)=0$. Similarly, $H(\RG|\RY,\RZ=\Rz)=0$. Thus, for a fixed $\RZ =\Rz$, $\RG$ is a c.f.\ of rvs $\RX$ and $\RY$ jointly distributed according to $P_{\RX, \RY \mid \RZ=\Rz}$. In this case, let $\op{mcf}(\RX,\RY)_{\RZ=\Rz}$ to denote the m.c.f. which indeed depends on the conditional distribution.  Because of the independence, $P_{\RX, \RY \mid \RZ=\Rz} =P_{\RX, \RY}$, however, the $\op{mcf}(\RX,\RY)_{\RZ=\Rz}$ remains same across all $\Rz$, and is equal to  $\op{mcf}(\RX,\RY)$. Therefore, from the optimality of m.c.f., we have $H(\RG |\RZ=\Rz) \leq  H(\op{mcf}(\RX,\RY)_{\RZ=\Rz} |\RZ=\Rz)=H(\op{mcf}(\RX,\RY) |\RZ=\Rz)= H(\op{mcf}(\RX,\RY))$, where the last equality follows from the independence of $\RZ$ and $(\RX,\RY)$. As a consequence, we have $H(\RG |\RZ) =\sum_{\Rz \in \mc{Z}}P_{\RZ}(\Rz) H(\RG|\RZ=\Rz)\leq H(\op{mcf}(\RX,\RY))$. The desired inequality follows from $H(\RG) \leq H(\RG,\RZ) =H(\RG |\RZ) + H(\RZ) \leq  H(\op{mcf}(\RX,\RY)) + H(\RZ) =H(\op{mcf}(\RX,\RY),\RZ)$.  This proves that $\op{mcf}((\RX,\RZ), (\RY,\RZ)) = (\op{mcf}(\RX,\RY), \RZ)$.

 \subsection{Useful Lemmas related to the proof of Theorem~\ref{thm:cwsk:irred}}
\begin{lemma}[Schwartz-Zippel lemma]\label{lem:sz}
Let $\op{P}(\RX_1,\ldots,\RX_n)$ be a non-zero polynomial in $n$ variables with degree $d$ and coefficients from a finite field $\Fq$. Given a non-empty set $S \subseteq \Fq$, if we choose the $n$-tuple $(\RMx_1, \ldots, \RMx_n)$ uniformly from $S^n$, then
\begin{align*}
\Pr \{(\RMx_1, \ldots, \RMx_n)\in S^n: \op{P}(\RMx_1, \ldots, \RMx_n) = 0\} \leq \frac{d}{|S|}.
\end{align*}
\end{lemma}

Fix two positive integers $m$ and $s$ such that $s\leq m$. Consider the integral domain $\Fq\left[\RX_{11}, \ldots ,\RX_{1m},\ldots, \RX_{s1}, \ldots ,\RX_{sm}\right]$, which is the set all multivariate polynomials in indeterminates $ \RX_{11}, \ldots ,\RX_{1m},\ldots, \RX_{s1}, \ldots ,\RX_{sm}$ with coefficients from a finite field, $\Fq$. Let us consider a matrix of the form
\begin{align}
\MM=\begin{bmatrix}
\op{L_1}(\RY_1)&\op{L_2}(\RY_1)&\cdots &\op{L_s}(\RY_1) \\
\op{L_1}(\RY_2)&\op{L_2}(\RY_2)&\cdots &\op{L_s}(\RY_2) \\
\vdots & \vdots & \ddots & \vdots\\
\op{L_1}(\RY_s)&\op{L_2}(\RY_s)&\cdots &\op{L_s}(\RY_s)
\end{bmatrix}_{s \times s}, \label{eqn:detmatrix}
\end{align}
where $\RY_k:=[\RX_{k1}, \ldots ,\RX_{km}]$ for $1 \leq k \leq s$ and $\op{L}(\RY_k)$  denotes a linear combination of indeterminates $ \RX_{k1}, \ldots ,\RX_{km}$ over $\Fq $. Note that row $k$ depends only on $\RY_k$.  Let  $\RX := [\RY^T_1, \ldots, \RY^T_s]^T$ and $\op{P}(\RX)$ denotes a polynomial in indeterminates $ \RX_{11}, \ldots ,\RX_{1m},\ldots, \RX_{s1}, \ldots ,\RX_{sm}$ with coefficients from $\Fq$. 
 
It is a fact \cite[p.~528]{bourbaki1989algebra} that  for a general matrix $\MM$ with entries from $\Fq\left[X\right]$, $\det(\MM)=0$ if and only if  there exist polynomials $\op{P_k}(\RX) $, $1 \leq k \leq s$, not all zero  such that
\begin{align*}
\MM\left[ \op{P_1}(\RX)  , \ldots ,  \op{P_s}(\RX) \right]^T= \M0.
\end{align*}
But this does not guarantee a non-zero $\lambda = [\lambda_1, \ldots, \lambda_s]  \in \Fq^s$ such that $ \MM \lambda^T= 0$.  However the following lemma shows that if the matrix is of the form  (\ref{eqn:detmatrix}), then this is the case.

\begin{lemma} \label{lem:det}
Let $\MM$ be matrix of the form (\ref{eqn:detmatrix}). Then  $\det(\MM)=0$ iff  there exists a  non-zero $\lambda = [\lambda_1, \ldots, \lambda_s]  \in \Fq^s$ such that $\MM \lambda^T= 0$.
\end{lemma}
\begin{proof}
The "if" part holds for any matrix $\MM$ by the fact stated above. 
For the "only if" part, suppose that $\det(\MM)=0$.  We can write $\MM$ as follows
\[\MM=\underbrace{\begin{bmatrix}
\RX_{11}&\RX_{12}&\cdots &\RX_{1m} \\
\RX_{21}&\RX_{22}&\cdots &\RX_{2m} \\
\vdots & \vdots & \ddots & \vdots\\
\RX_{s1}&\RX_{s2}&\cdots &\RX_{sm} 
\end{bmatrix}}_{=\MX}\underbrace{\begin{bmatrix}
a_{11}&a_{21}&\cdots &a_{s1} \\
a_{12}&a_{22}&\cdots &a_{s2} \\
a_{13}&a_{23}&\cdots &a_{s3} \\
\vdots & \vdots & \ddots & \vdots\\
a_{1m}&a_{2m}&\cdots &a_{sm} 
\end{bmatrix}}_{:=\MA}.\]
 for some $\MA \in \Fq^{m \times s}$.  
Now consider the determinant of the matrix $\MM$,
\begin{align*}
&\det(\MM) = \sum_{\sigma \in S_s} \sgn(\sigma )\op{L_{\sigma(1)}}(\RY_1)\ldots \op{L_{\sigma(s)}}(\RY_s)\\
&=\sum_{\sigma \in S_s}\sgn(\sigma ) \left( \sum_{j_1=1}^{m}a_{\sigma(1)j_1} \RX_{1j_1}\right)\ldots \left( \sum_{j_s=1}^{m}a_{\sigma(s)j_s} \RX_{sj_s}\right)\\
&= \sum_{\sigma \in S_s}\sgn(\sigma )\sum_{j_1,\ldots,j_s \in [m]^s}\left(a_{\sigma(1)j_1}\ldots a_{\sigma(s)j_s}\right)\RX_{1j_1}\ldots \RX_{sj_s}\\
& \stackrel{(a)}{=} \sum_{j_1,\ldots,j_s\in [m]^s}\left(\sum_{\sigma \in S_s}\sgn(\sigma )a_{\sigma(1)j_1}\ldots a_{\sigma(s)j_s}\right)\RX_{1j_1}\ldots \RX_{sj_s}\\
&\stackrel{(b)}{=} \sum_{j_1,\ldots,j_s\in [m]^s} \det(A_{j_1\ldots j_s})\RX_{1j_1}\ldots \RX_{sj_s} 
\end{align*}
where $\MA_{j_1j_2\ldots j_s}$ is the $s\times s$ submatrix of $\MA$ formed by the rows $j_1, j_2, \dots ,j_s$. $(a)$ follows from the fact that the monomials $\RX_{1j_1}\RX_{2j_2}\ldots \RX_{sj_s}$, for $j_1,j_2,\ldots,j_s\in [m]^s$, are distinct. $(b)$ holds because the inner sum is just the determinant of $\MA_{j_1j_2\ldots j_s}$. Since $\det(\MM)=0$, $\det(\MA_{j_1j_2\ldots j_s})= 0$ for every distinct indices $j_1,j_2,\dots,j_s$ which implies that  any $s$ rows of $\MA$ are linearly dependent over $\Fq$. This shows that the rank$_{ \Fq}(\MA) < s$, therefore the columns of $\MA$ are linearly dependent over $\Fq$. Hence there exists a  non-zero $\lambda = [\lambda_1, \ldots, \lambda_s]  \in \Fq^n$ such that $ \MA\lambda^T= 0 \Rightarrow \MM\lambda ^T= 0$.
\end{proof}

\begin{definition}
 Let $\MW$ be a row-partitioned matrix of the form  
 \begin{align} \label{blockcolumn}
 \renewcommand{\arraystretch}{1.5}
 \begin{bmatrix}
 \begin{array}{c}
  \MW_1\\ \hline
  \MW_2\\ \hline
  \vdots\\ \hline
  \MW_{|E|}
  \end{array}
  \end{bmatrix}
 \end{align}
where $\MW_i$ is a $n_i \times n_w$ matrix over $\Fq$. Then we say that the matrix $\MW$ is \emph{reducible} if there exist an index $i$ and a non-zero row vector $r_i$ in $\Fq^{(n_i)}$ such that column span of $\MW$ contains the column vector $[-0- \mid  \cdots \mid - r_i-\mid \cdots \mid -0-]^T$. If the matrix $\MW$ is not reducible then we say it is \emph{irreducible}
\end{definition}
A tree-PIN source with linear wiretapper is irreducible iff  the wiretapper matrix $\MW$ is irreducible.
\begin{lemma} \label{lem:upbdirred}
 Given a $(\sum_{e \in E} n_e) \times  n_w$ wiretapper  matrix $\MW$ in the row-partitioned form \eqref{blockcolumn}. If the matrix is irreducible then $n_w  \leq (\sum_{e \in E}n_e)-s$ where $s=\min\{n_e: e \in E\}$. 
\end{lemma}
\begin{proof}
 By  elementary column operations and block row swapping, we can reduce $\MW$ into the following form
 \begin{align*}
 \renewcommand{\arraystretch}{1.5}
 \left[\begin{array}{cccc}
 \MW_{11}&\M0& \cdots&\M0\\ \hline
 \MW_{21}&\MW_{22}& \cdots &\M0\\ \hline
 \vdots&\vdots&\ddots&\vdots\\ \hline
  \MW_{k1}&\MW_{k2}& \cdots &\MW_{kk}\\ \hline
  \vdots&\vdots&\ddots&\vdots\\ \hline
  \MW_{|E|1}&\MW_{|E|2}& \cdots &\MW_{|E|k}\\
 \end{array}\right]
\end{align*}
 where the diagonal matrices $\MW_{jj}$ are full-row rank matrices. The  upper bound on $k$ is $(|E|-1)$, because of the irreducibility. The upper bound on  the number of columns in $\MW_{jj}$ is  $n_{e_j}$, where $e_j$ is the edge corresponding to the row $j$ (after block row swapping). So, 
 \begin{align*}
  n_w & \leq \max\biggl\{\sum_{j\in K} n_{e_j}: K \subseteq [|E|], |K| \leq (|E|-1)\biggr\} \\
  &\leq \max\biggl\{\sum_{j\in K} n_{e_j}:  |K| = (|E|-1)\biggr\}\\
  &= \max\biggl\{\sum_{e\in E} n_e - n_{e'}:  e' \in E\biggr\}\\
  &=\sum_{e\in E} n_e - s
 \end{align*}

\end{proof}

\begin{lemma}\label{lem:nonzeropoly}
 Given a $(\sum_{e \in E} n_e) \times  n_w$ wiretapper  matrix $\MW$ with full column rank such that $n_w  \leq (\sum_{e \in E}n_e)-s$ where $s=\min\{n_e: e \in E\}$. Let $\MS\MW =0$, where $\MS=(\MS_e, \MT_e)_{e \in E}$, where $\MS_e$ is an $s \times s$ matrix and $\MT_e$ is an $s \times (n_e -s)$ matrix. Then  if $\MW$ is irreducible then $\prod_{e\in E} \det (\MS_e)$ is a non-zero polynomial (Polynomial in terms of the inderterminates corresponding to the free variables of $\MS$  corresponding to $\MS\MW=0$). 
\end{lemma}
\begin{proof}
 Suppose  $\prod_{e\in E} \det (\MS_e)$ is a zero polynomial then $ \det (\MS_i) \equiv 0$ for some $i\in E$. Let $m := \sum_{e \in E}n_e - n_w $, it follows from lemma \ref{lem:upbdirred} that $m \geq s$. Since $\MS$ satisfies $\MS\MW = 0$, in each row of $\MS$ there are $m$ independent variables, which are indeterminates, and every other element in the row is expressed  as  a linear combination of these indeterminates. So, in total there $sm$ indeterminates in $S$ , without loss of generality, assume them to be in the first $m$ columns of $\MS$. Now $\MS_i$ has the form similar to $\eqref{eqn:detmatrix}$ for some linear functions. From lemma \ref{lem:det}, $ \det (\MS_i) \equiv 0$ implies that there exists a non-zero $\lambda = [\lambda_1, \ldots, \lambda_s]  \in \Fq^s$ such that $ \MS_i \lambda^T= 0$. Consider the block column partitioned  row vector $\MR$ such that the block corresponding to the edge $i$ is  $\MR_i = [\lambda_1, \ldots, \lambda_s, 0,\ldots, 0]$ and  $\MR_j =[- 0 -]$. Then $ \MS \MR^T= 0$. Consider the matrix $\tMW = [\MW \mid \MR^T]$ which also satisfies $\MS \tMW =\M0$. One can see that $\ker(\tMW^T) \subseteq \ker(\MW^T) $. For the other direction, note that any vector in the $\ker(\MW^T)$ also belongs to $\ker(\MR^T)$. As a consequence $\ker(\tMW^T) = \ker(\MW^T) $, then the dimension of the column space of $\tMW$ is $\sum_{e\in E} n_e  - \dim(\ker(\tMW^T)) =\sum_{e\in E} n_e  - \dim(\ker(\tMW^T)) = n_w$. Hence $\MR^T$ is in the column span of $\MW$ which implies that $\MW$ is reducible.  \end{proof}

\fi

\end{document}